\documentclass[12pt]{article}
\usepackage{amssymb}
\usepackage{amsmath}
\usepackage{amsthm}
\usepackage{graphicx}
\usepackage{subfigure}
\usepackage[utf8]{inputenc}
\usepackage[T1]{fontenc}
\usepackage[english]{babel}
\usepackage[square,numbers]{natbib}
\usepackage{float}
\usepackage{caption}
\usepackage{hyperref}
\usepackage[left=1in, right=1in, top=1in]{geometry}
\usepackage{color,soul}
\usepackage{authblk}
\newtheorem{theorem}{Theorem}

\begin{document}
\title{\textbf{Two-Phase Modeling of Fluid Injection Inside Subcutaneous Layer of Skin}}
\author[1]{Abdush Salam Pramanik\thanks{Electronic address: rs\_abdush@nbu.ac.in}}
\author[1]{Bibaswan Dey\thanks{Electronic address (Corresponding Author): bibaswandey@nbu.ac.in}}
\author[2]{Timir Karmakar\thanks{Electronic address: tkarmakar@nitm.ac.in}}
\author[1]{Kalyan Saha\thanks{Electronic address: kalyansaha@nbu.ac.in}}
\affil[1]{Department of Mathematics, University of North Bengal, Raja Rammohunpur, Darjeeling-734013, West Bengal, India}
\affil[2]{Department of Mathematics, National Institute of Technology Meghalaya, Shillong-793003, Meghalaya, India}
\date{}
\maketitle
\begin{abstract}
\noindent
The motivation of this present theoretical investigation comes from the delivery of various drugs and vaccines through subcutaneous injection (SCI) to the human body. In the SCI procedure, a medical person takes a big pinch of skin of the injection applicable area of a patient to pull the subcutaneous layer (SCL) from the underlying muscle layer to smooth the execution of the procedure. Generally, the human skin, particularly SCL, is a heterogeneous and anisotropic living material. The major constituents of the SCL are adipose cells or fat cells and interstitial fluid. These adipose cells are oriented in such a way that the hydraulic conductivity of the SCL exhibits anisotropy. Consequently, one can adopt the field equations of mixture theory to describe the continuum nature of SCL. This mathematical modeling involves a perturbation technique where the small aspect ratio of the SCL provides a valid perturb parameter. This study highlights the issue of the mechanical response of the adipose tissue in terms of the anisotropic hydraulic conductivity variation, the viscosity of the injected drug, the mean depth of subcutaneous tissue, etc. In particular, the computed stress fields can measure the intensity of pain experienced by a patient after this procedure. Also, this study discusses the biomechanical impact of the creation of one or more eddy structure(s) due to the high pressure developed, increased tissue anisotropy, fluid viscosity, etc., within the area of applying injection.
\end{abstract}

\section{Introduction}
\noindent
Drug injection is a popular and efficient way to deliver a drug into biological tissues in order to get more appropriate results. Among the several injection techniques, subcutaneous injection (SCI) is a useful as well as highly effective corresponding to the medication of insulin, morphine, epinephrine, hydromorphone, diacetylmorphine, goserelin, metoclopramide, heparin, fertility drugs etc. inside fatty subcutaneous tissue immediately below the dermis layer (DL) \citep{kim2017effective}. SCI becomes advantageous for possible self-administration for the patients who need particular medicine regularly \citep{ogston2014subcutaneous}. Also, this technique becomes an alternative way of drug intake that results in better drug mobility for patients with poor venous access \citep{dychter2012subcutaneous,shapiro2013use}. In the context of safety and efficacy, the SCI is better than any other technique, such as intravenous or intramuscular injection \citep{allen2013ansel,prettyman2005subcutaneous}. A survey done by \citet{stoner2015intravenous} suggests that the subcutaneous route is preferable to the patients as compared to the intravenous route. For those patients who need multiple daily doses of one or more drug(s), SCI provides a broader range of alternative sites of injection than intramuscular injection \citep{haller2007converting}.\\

\noindent
To understand the detailed mechanism of SCI and the corresponding mechanical response of the tissue within the area of application of injection, it is our primary goal to understand the composition of the tissue at the injection site. The subcutaneous layer (SCL) is in general a composition of adipose tissue along with extracellular fluid \citep{geerligs2008linear,comley2011deep,derler2012tribology}. According to  \citet{shrestha2018fluid,shrestha2020imaging}, skin tissue behaves as a deformable porous medium that absorbs fluid as a result of the formation of a cavity under the local expansion of tissue rather than rupturing.\\

\noindent
The fluid flow through the tissue matrix during an intradermal injection is affected by its porosity and permeability. Hence, fluid flow and deformation of solid phases get coupled \citep{barry1992flow}. The porosity and permeability variation during an injection plays a vital role in controlling the accuracy of the amount of fluid injected into the skin at different stages of the injection \citep{shrestha2018fluid}. Consequently, one can control the dosage of a drug to be delivered. Identifying the field equations that govern the above phenomena would be essential. In this context, one can go through the classical study of \citet{oomens1987mixture} where skin tissue has been considered as a biphasic mixture of solid ($s$) and fluid ($f$) constituents. A set of non-linear field equations can describe the biphasic nature of skin tissue. In this context, one can go through the classical literature on field theories of mechanics for a detailed structure of governing field equations \mbox{\citep{truesdell1960classical,rajagopal1995mechanics,truesdell2004non,rajagopal1986boundary,gandhi1987some,wineman1992shear}}. For a wide range of authorship, those field equations of mixture theory are presented in some next-generation literature. In this study, we maintain adequate clarity to express the principal aspect of the modelling using the fundamental laws of mechanics, e.g., mass and momentum conservation equations. If $\rho_i$ and $\mathbf{V}_{i}$ represent apparent mass density and velocity of $i^{\textrm{th}}$ constituent ($i\in\{s,~f\}$) then the mass balance equation for $i^{\textrm{th}}$ constituent becomes
\begin{equation}\label{IEq1}
\frac{\partial \rho_i}{\partial t} + \nabla. (\rho_i \mathbf{V_i})=m_i,\hspace{0.5cm} i\in\{s,~f\},
\end{equation}
\begin{equation}\label{IEq2}
\sum_{i=s,f} m_{i} = 0.
\end{equation}
On the other hand, the balance of momentum for the $i^{\textrm{th}}$ constituent is given by
\begin{equation}\label{IEq3}
\rho_i \frac{D\mathbf{V}_i}{Dt} = \nabla . \mathbf{T}_i + \rho_i \mathbf{F}_i + \boldsymbol{\Pi}_i,    \hspace{0.5cm} i\in\{s,~f\},
\end{equation}
where ${D\mathbf{V}_i}/{Dt}= {\partial}/{\partial t} + (\nabla. \mathbf{V}_i)$ denotes
 the material derivative, $\mathbf{T}_i$ represents cauchy's stress tensor, $\mathbf{F}_i$ is the external body force corresponding to the $i^{\textrm{th}}$ constituent and $\boldsymbol{\Pi}_i$ is the interactive force on $i$-th constituent due to the other. In addition, the balance of momentum of the whole tissue matrix leads to
\begin{equation}\label{IEq4}
\sum_{i=s,f}(\boldsymbol{\Pi}_i + m_i \mathbf{V}_i) = 0.
\end{equation}
Similar sets of equations (\ref{IEq1})-(\ref{IEq4}) have been reported in several studies based on mixture theory \citep{rajagopal1995mechanics,rajagopal1986boundary,gandhi1987some,wineman1992shear,byrne2003modelling,rajagopal2007hierarchy,kumar2018nutrient,dey2021mathematical,anguiano2022mixture}. Depending on the physical behavior, for an appropriate modelling, a tissue can be considered as mixture of two or several fluid constituents. For example, while considering growth of a tissue, one can assume the cellular phase to behave as a fluid continuum \citep{byrne2003two,breward2003multiphase,rajagopal2010mechanics}. Therefore, the dynamics of each constituent can be governed by a similar fluid momentum equation. Each constituent can be distinguished from others in terms of viscosity \citep{chen2001influence,byrne2003modelling}. However, the situation would be different when flow-induced deformation of biological tissues is studied. \citet{barry1990comparison} compared the flow-induced deformation between soft biological tissues and polyurethane sponge through a mathematical model assuming the solid phase to behave as a poroelastic material. In this regard, the dynamics of the whole tissues are governed by the set of equations (\ref{IEq1})-(\ref{IEq4}) stated above. But the Cauchy stress tensor ($\mathbf{T}_s$) corresponding to the solid phase has to follow the stress relation for an elastic material. On the other hand, the fluid stress ($\mathbf{T}_f$) can depend either solely on the pore pressure \citep{barry1990comparison} or both the pore pressure and fluid viscous force \citep{barry1991fluid}. Note that any volumetric change in the tissue due to fluid-induced deformation is infinitesimal. \citet{barry1995injection} studied fluid injection as a point source into a deformable porous layer with both the boundaries impermeable to fluid flow using biphasic mixture theory. Later, \citet{barry1997deformation} extended their study by considering a set of boundary conditions where the upper surface is permeable to fluid flow. In this context, the models of \citet{li2009mathematical} are relevant to the SCI of insulin. Recently, \citet{shrestha2020imaging} have introduced a semi-empirical model based on experimental data and constitutive equations of flow through biological tissue that elucidates the fluid transport through skin tissue. Injecting fluid into the skin involves a coupled interaction with the deformation of the soft porous matrix of skin tissue since skin tissue is a deformable porous medium. Their study assumes the permeability ($K$) of any biological tissue (particularly the skin tissue), in the following form $K=K_{0}\exp(\mathcal{M}\mathcal{E}_{ij})$. $\mathcal{E}_{ij}$ denotes the volumetric strain in the solid matrix of skin tissue; $K_{0}$ and $\mathcal{M}$ are macroscopic material constants. This form of permeability is general; however,  other forms of permeability can also be discussed. The concept of anisotropic permeability comes when $K_{0}$ assumes a form of a square matrix, i.e., $K_{0}=K_{0_{ij}}$ when $ij$th entry is the permeability in the $ij$th direction concerning the coordinate axes. For simplicity, $\mathcal{M}$ can be chosen zero as a particular case. \\

\noindent
Besides the composition of tissue structure, tissue hydraulic conductivity or tissue permeability plays an important role in delivering the drug through the injection. If we consider a soft biological tissue as a deformable porous media \citep{barry1992flow}, it may consist either of an isotropic matrix whose permeability is the same along all directions \citep{mow1984fluid,holmes1985theoretical,wei2003flow,dey2016hydrodynamics,shrestha2020imaging}, or an anisotropic matrix whose permeability varies with direction \citep{rees1995effect,payne2001effect,karmakar2018effect,rajani2020anisotropic}. In particular, the anisotropic permeability may varies in the principal directions only \citep{kohr2008green,karmakar2017note}. The effects of anisotropic permeability have been observed in the various study of articular cartilage \citep{reynaud2006anisotropic,federico2008anisotropy,iatridis1998degeneration}. \citet{federico2008permeability} studied the effects of anisotropic permeability in a biological tissue filled with interstitial fluid and reinforced by impermeable collagen fibers. Most of the previous studies are mainly based on the isotropic nature, and there are few involving the anisotropic behavior of biological tissue. But due to the variations in the distribution of collagen fibers, soft connective tissue can show anisotropic behavior \citep{holzapfel2001biomechanics}. SCL is a soft connective layer of tissue; it possesses anisotropic permeability. As reported by \citet{kim2017effective} for a fixed flow rate, vertical permeability of skin tissue is greater than horizontal permeability, and there is no strong evidence that the converse may not hold. This incident motivates us to think about the situation when the horizontal permeability is greater than the vertical permeability. Therefore in this study, we consider the anisotropic nature of the SCL in the case mentioned above, which may be an interesting topic.\\

\noindent
A detailed literature review indicates the lacuna in mathematical modeling of fluid flow problems inside soft connective tissues, including both anisotropic and deformable nature. Consequently, this article introduces a mathematical model to discuss drug delivery through an SCI. The subcutaneous tissue region has a biphasic description of two main constituents fat adipose cells (AC) and interstitial fluid (IF). In addition, the interstitial hydraulic conductivity has anisotropic, which varies in the principal directions only. In this current study, our primary aim is to discuss the mechanical response of the SCL in terms of the variation of anisotropic hydraulic conductivity, the viscosity of the injected drug, the mean depth of the SCL, etc. In addition, we would like to study the pain realized by a patient near the injection site with the help of pressure gradient and shear stress.
\begin{figure}[h!]
\centering
\includegraphics[scale=0.6]{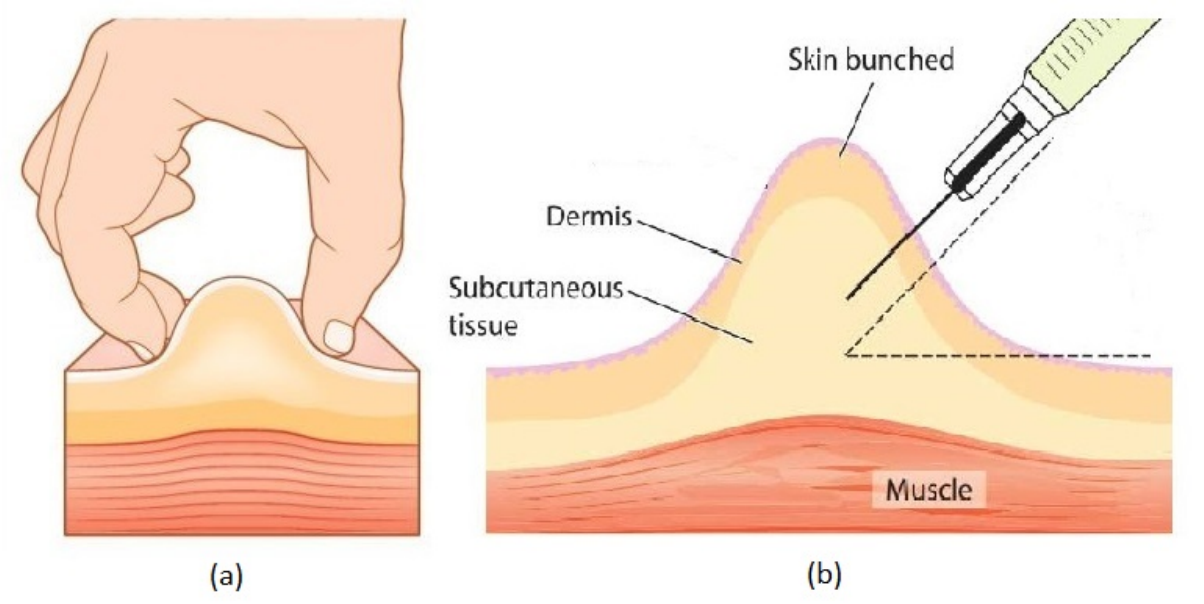}
\caption{Cartoon diagram of subcutaneous injection (SCI): (a) Skin is bunching during injection, (b) Needle injecting after skin bunched \citep{shepherd2018injection}.}
\label{Carton}
\end{figure}
\section{Mathematical formulation}
\noindent
SCIs are typically used to administer drugs and vaccines into the fatty tissue layer (subcutaneous tissue) sandwiched between the dermis layer (DL) and muscle layer (ML). While injecting a fluid containing drugs or vaccines into a patient, the lifted skin fold technique must be used to avoid the risk of muscle damage. The best method is to lift the skin of the injection site to pull the fat tissue within the SCL away from the underlying ML and hold it for the entire duration of the injection procedure (Fig. \ref{Carton}).\\

\begin{figure}[h!]
\centering
\includegraphics[scale=0.7]{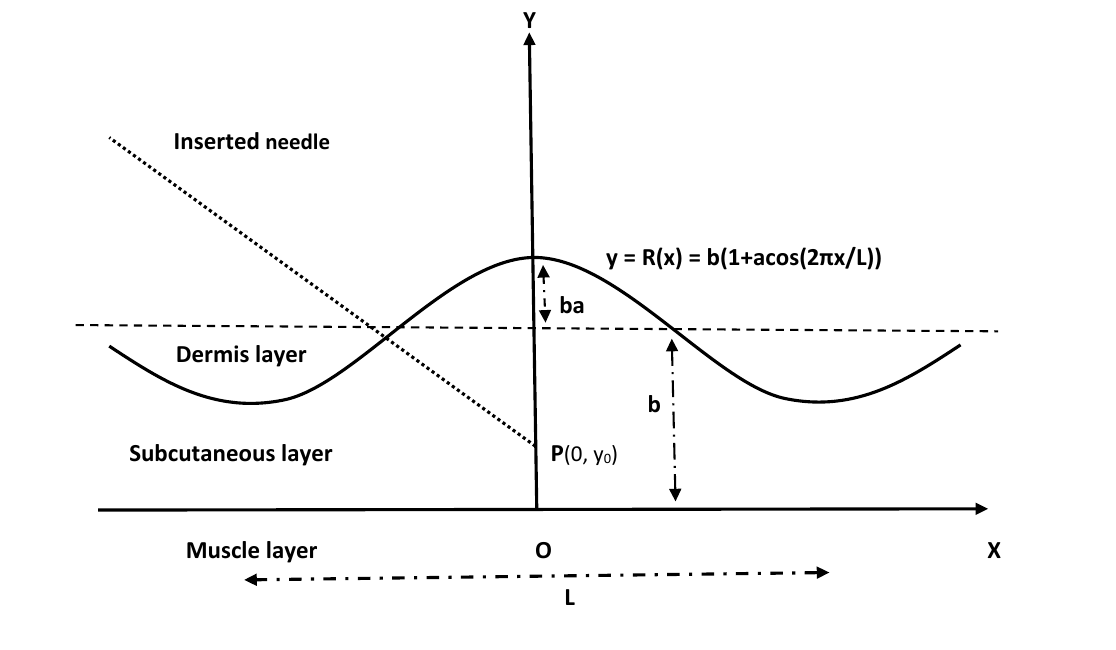}
\caption{Schematic of the mathematical model approximating the fluid injection process.}
\label{Geometry}
\end{figure}
\noindent
Referring to Fig. \ref{Geometry}, $(0,y_0)$ is considered as the point of injection, i.e., the position of the tip of the needle. Therefore, $b-y_{0}$ is the penetration depth from the interface of SCL and DL (may be regarded as SD interface). The SCL is initially bounded by a permeable upper DL (before injection) located at $y=b$ and a permeable lower ML at $y=0$. If a medical staff takes a big pinch of the patient's skin using the thumb and index fingers and holds, the SCL gets pulled away from the ML to make the injection easier. Consequently, the interface of SCL and DL (i.e., SD interface) assumes the shape of a cosine curve of the form $y=R(x)=b\left(1+a\cos(2\pi x/L)\right)$ within $x\in[-L/2, L/2]$ with $a$ as its amplitude such that $ba$ signifies the increased height of the SCL as a result of skin lifting.\\

\noindent
The present study mainly deals with the impact of tissue anisotropy on the motion of injected fluid rather than the time duration of injection. Consequently, one may assume the SCI process described here is time independent. Two major components of the SCL are IF and AC. In general, the cells within SCL are oriented in such a pattern that the overall permeability becomes anisotropic. Subsequently, the SCL can be considered as a deformable porous media or poroelastic media with anisotropic permeability where the cellular phase acts as solid material and the governing equations of biphasic mixture theory can describe the fluid mechanical process within SCL. However, there will be an issue regarding decoupling of one phase from the other in the governing equations under the assumption of steady state \citep{dey2016hydrodynamics,dey2016theoretical}. It must be emphasized that the momentum equation for the fluid component does not take into account the solid displacement term. In contrast, the solid constituent's momentum equation does have the fluid velocity term. This may be regarded as a one-way coupling between two principal constituents present within the SCL. In order to avoid such situation, both IF and AC can be assumed to be in the fluid phase with different viscosities \citep{byrne2003two,breward2003multiphase}. After successful adminstration of a drug loaded fluid through SC injection, the injected fluid becomes a part of the IF as we consider both of them to have a similar dynamic viscosity. If $\mu^{(\textrm{IF})}_{f}$ and $\mu^{(\textrm{D})}_{f}$ are the viscosities, $\varphi^{(\textrm{IF})}_{f}$ and $\varphi^{(\textrm{D})}_{f}$ are the volume fractions of IF and injected drug-loaded fluid respectively, the effective viscosity of IF in presence of injected drug loaded fluid becomes $\mu_{f}\thickapprox \varphi^{(\textrm{IF})}_{f}\mu^{(\textrm{IF})}_{f} + \varphi^{(\textrm{D})}_{f}\mu^{(\textrm{D})}_{f}$. On the other hand, $\mu_c$ denotes the viscosity of the adipose cell.\\

\noindent
A two-dimensional steady motion of these two incompressible fluids within the SCL is considered. SCL does not allow fast absorption of an injected drug due to the presence of fewer blood vessels \citet{barbieri2014skin}. Also, within the fat tissue of SCL, the interstitial gap is expected to be small. Consequently, the motions of two incompressible fluid constituents (IF and AC) are governed by Brinkman type equations where the interactions between constituent phases are involved and given by \citep{rajagopal2007hierarchy,dey2016hydrodynamics}
\begin{equation}\label{MEq1}
-\varphi_f \nabla P + \left(\lambda_f + \mu_f\right) \nabla (\nabla. \mathbf{u_f}) + \mu_f \nabla^2 \mathbf{u_f} - \mu_f \mathbf{K}^{-1}(\mathbf{u}_f-\mathbf{u}_c)=0,
\end{equation}
\begin{equation}\label{MEq2}
-\varphi_c \nabla P + \mu_c \nabla^2 \mathbf{u_c} + \mu_f \mathbf{K}^{-1}(\mathbf{u_f}-\mathbf{u_c})=0,
\end{equation}
In addition, above momentum equations are supported by the following equations of mass conservation:
\begin{equation}\label{MEq3}
\nabla.(\varphi_f \mathbf{u_f})=F(x,y),
\end{equation}
and
\begin{equation}\label{MEq4}
\nabla.(\varphi_c \mathbf{u_c})=0,
\end{equation}
where $\mathbf{u}_{f}=(u_{f}, v_{f})$ and $\mathbf{u_c}=(u_{c}, v_{c})$ are the velocity vector for IF and AC respectively; $\varphi_f$ and $\varphi_c$ are the volume fraction for IF and AC respectively with $\varphi_f + \varphi_c = 1$; $F(x,y)$ represents the source corresponding to the fluid injection; $P$ is the hydrodynamic pressure. If we drop adipose cell velocity and IF velocity from Eqs. (\ref{MEq1}) and (\ref{MEq2}) respectively, we get Brinkman equation (governing momentum equations for flow through a rigid porous media) \citep{karmakar2017note,savatorova2011homogenization,dey2014mass} in each case. Note that the second term on the left-hand side of Eq. \mbox{(\ref{MEq1})} arises due to the nonzero source term $F(x,y)$ in the mass conservation equation \mbox{(\ref{MEq3})} corresponding to IF. However, in Eq. \mbox{(\ref{MEq2})} no such term arises due to the absence of sources in the mass conservation equation \mbox{(\ref{MEq4})} for AC since the proliferation of AC is neglected. \\

\noindent
Considering the arbitrary orientation of fat and connective tissues within the SCL, the permeability tensor $\mathbf{K}$ possesses both non-zero off-diagonal entries along with dissimilar diagonal elements. In other words, $\mathbf{K}$ may depend on the anisotropic angle $\phi$ between the horizontal direction and the principal axis \citep{degan2002forced,karmakar2018effect}. The scenario at hand bears striking similarities to the flow through a porous material that is both deformable and anisotropic. Such materials typically exhibit three distinct sets of directions, including the principal axes of stress and strain, which possess clear and unambiguous physical properties. Detailed anisotropic description of such solid materials can be found in \mbox{\citet{truesdell1960classical,truesdell2004non}}. Our study solely concentrates on a basic scenario where the permeability tensor $\mathbf{K}$ exclusively comprises components that align with the principal directions \mbox{\citep{kohr2008green,karmakar2017note}}:
\begin{equation}
\mathbf{K} =
\begin{bmatrix}
K_1 & 0\\
0 & K_2
\end{bmatrix},
\end{equation}
with $K_1$ and $K_2$ are the permeabilities along the $x$ and $y$ directions (i.e. in principal directions) respectively.
\subsection{Boundary conditions:}
\noindent
In order to proceed further and derive the solution of the problem, we consider the following boundary conditions:\\

\noindent
(i) At the interface of SCL and DL (SD interface), denoted as $y=R(x)$, the tangential and normal velocity components are assumed as follows:
\begin{equation}\label{BC1}
 \mathbf{u}_{f}.\hat{\mathbf{t}}= 0~~~\textrm{and}~~~\mathbf{u}_{f}.\hat{\mathbf{n}}= V_D(x),
\end{equation}
where $\hat{\mathbf{t}}$ and $\hat{\mathbf{n}}$ are the unit tangent and normal vector respectively on $y=R(x)$. This condition implies that the tangential velocity of IF vanishes at the SD interface, while the normal velocity is equal to the vertical permeation $V_{D}(x)$ through the interface. On the other hand, at $y=R(x)$, both the velocity components of AC are set to zero:
\begin{equation}\label{BC2}
 u_{c}=0~~~\textrm{and}~~~v_{c}=0.
\end{equation}
(ii) At $y=0$, corresponding to the interface between the SCL and the ML (which may be regarded as the SM interface), the horizontal (tangential) and vertical (normal) velocity components are assumed to follow as
\begin{equation}\label{BC3}
 u_{f} = \lambda_{s_{f}} \frac{\partial u_f}{\partial y} \text{    ,    } v_{f} = V_M(x),
\end{equation}
\begin{equation}\label{BC4}
 u_{c}=\lambda_{s_{c}} \frac{\partial u_c}{\partial y} \text{     and    } v_{c}=0 .
\end{equation}
where $\lambda_{s_{f}}$ and $\lambda_{s_{c}}$ are the slip coefficients.\\

\noindent
The injected fluid (containing drug molecules) dissolved within the interstitial fluid of the SCL can undergo vertical permeation through both the SD and SM interfaces, determined by $V_{D}(x)$ and $V_{M}(x)$ respectively. However, it is reasonable to consider that the ACs within the SCL are static on both the SD and SM interfaces so that they cannot permeate through due to their high density and size. These assumptions are reasonable based on the physical characteristics of adipose cells \mbox{\citep{stenkula2018adipose}}. Therefore, both the SD and SM interfaces act as porous filter which permits IF to undergo a vertical permeation but restricts ACs. An equivalent situation is observed when fluid diffuses through an elastic solid undergoing large deformations \mbox{\citep{prasad2006diffusion}}. Note that corresponding to the constant flow rate, we can get $V_D(x)$ is equal to $V_M(x)$ ( see Theorem 1). The expressions of $V_D(x)$ or $V_M(x)$ are considered later.\\

\noindent
The structural difference between the SCL and ML suggests the permeability variation. According to \mbox{\citet{kim2017effective}} corresponding to the same flow rate in horizontal and vertical directions, the horizontal permeability of the SCL is found to be higher than that of the ML. On the other hand, during the injection procedure, a horizontal motion at the SM interface may be noted due to the variation of fluid pressure (generated from the injection site) in the axial direction. Such motion can be characterized by the boundary conditions Eq. \mbox{(\ref{BC3})-(\ref{BC4})} as proposed by \mbox{\citet{beavers1967boundary}}, \mbox{\citet{jones1973low}}, \mbox{\citet{karmakar2018squeeze}}, and \mbox{\citet{hill2008poiseuille}} at the SM interface are regarded as slip conditions where the parameter $\lambda_s$ (called slip coefficient) is directly proportional to the length scale same as $\sqrt{K_1}$. In particular, we consider $\lambda_{s_{f}}=\varphi_{f}\lambda_{s}$ and $\lambda_{s_{c}}=\varphi_{c}\lambda_{s}$ (where $\lambda_{s}$ is the common slip coefficient at the SD interface).\\

\noindent
(iii) Flux condition:
Let $Q$ be the volumetric flow rate across the region which is given by
\begin{equation}\label{BC5}
Q=\int_{0}^{R(x)} (\varphi_{f}u_{f} + \varphi_{c} u_{c}) \hspace{0.01 cm}dy .
\end{equation}
\subsection{Non-dimensionalisation}
\noindent
We introduce the following non-dimensional variables:\\
$x'=x/L$, $y'=y/b$, $\delta = b/L$, $u_{i}'=u_i/(Q/b)$, $v_{i}'=v_i/(Q/L)$, $p'=p/(\mu_f QL/K_{1}b),$ for $i=f, c$ in Eqs. (\ref{MEq1})-(\ref{MEq4}) to make them dimensionless. Accordingly, we obtain following dimensionless parameters those appear in the dimensionless governing equations: $\mu_r=\mu_f/\mu_c$, $\text{Da}=K_1/b^2$, $\lambda^2=K_1/K_2$, $V'_D=V_D/(Q/L)$, $V'_M=V_M/(Q/L)$, $\lambda'_s=\lambda_s/L$ and $F'=F/(Q/bL)$.\\

\noindent
Here, $\delta$ denotes the aspect ratio of the SCL; $\textrm{Da}$ is the ease of interstitial fluid percolation through the SC tissue in the horizontal direction (equivalent to the Darcy number); $\mu_r$ is the ratio of the IF to the AC and $\lambda^2$ is the ratio of horizontal permeability to the vertical permeability, which can be referred to as the anisotropic ratio.\\

\noindent
Correspondingly, the non-dimensional governing equations can be written as (omitting dash)\\
\begin{equation}\label{NEq1}
\delta^2\left(1+\frac{\lambda_{f}}{\mu_{f}}\right)\frac{\partial}{\partial x} \left(\frac{\partial u_f}{\partial x} + \frac{\partial v_f}{\partial y}\right) + \left( \delta^2 \frac{\partial^2 u_f}{\partial x^2} + \frac{\partial^2 u_f}{\partial y^2}  \right) - \alpha^2\left[\left(u_f - u_c\right)+ \varphi_{f}\frac{\partial p}{\partial x}\right]=0,
\end{equation}
\begin{equation}\label{NEq2}
 \delta^2\left(1+\frac{\lambda_{f}}{\mu_{f}}\right)\frac{\partial}{\partial y} \left(\frac{\partial u_f}{\partial x} + \frac{\partial v_f}{\partial y}\right) + \delta^2 \left( \delta^2 \frac{\partial^2 v_f}{\partial x^2} + \frac{\partial^2 v_f}{\partial y^2}  \right) - \alpha^{2}\left[\delta^2 \lambda^2 \alpha^2(v_f - v_c)+\varphi_{f} \frac{\partial p}{\partial y}\right]  = 0,
\end{equation}
\begin{equation}\label{NEq3}
\left(\delta^2\frac{\partial^2 u_c}{\partial x^2}+\frac{\partial^2 u_c}{\partial y^2}\right)+ \left[\alpha^2\mu_r(u_f-u_c)-\varphi_c\frac{\partial p}{\partial x} \right]=0,
\end{equation}
\begin{equation}\label{NEq4}
\delta^2 \left (\delta^2 \frac{\partial^2 v_c}{\partial x^2}+\frac{\partial^2 v_c}{\partial y^2}\right)+ \alpha^2\mu_r\left[\delta^2 \lambda^2 (v_f-v_c)-\varphi_c\frac{\partial p}{\partial y}\right]=0,
\end{equation}
\begin{equation}\label{NEq5}
\left(\frac{\partial u_f}{\partial x} + \frac{\partial v_f}{\partial y}\right)= F(x,y),
\end{equation}
\begin{equation}\label{NEq6}
\left(\frac{\partial u_c}{\partial x} + \frac{\partial v_c}{\partial y}\right)= 0,
\end{equation}\\
where $\alpha^2=1/\text{Da}$. Here we consider Stokes hypothesis by taking $2\lambda_f + 3 \mu_f = 0$, i.e. ${\lambda_{f}}/{\mu_{f}}=-3/2$ \citep{dey2016hydrodynamics}.
Also the corresponding boundary conditions are (dropping dash) \\\\
(i) on $y=R(x)=1+a\cos(2\pi x),$\\
\begin{equation}\label{BC6}
u_{f}=2\pi a \delta^{2}\frac{\sin(2\pi x)}{\sqrt{1+4\pi^{2}a^{2}\delta^{2}\sin(2\pi x)}}V_D(x), \hspace{0.1cm} v_{f}= \frac{1}{\sqrt{1+4\pi^{2}a^{2}\delta^{2}\sin(2\pi x)}}V_D(x),
\end{equation}
\begin{equation}\label{BC7}
u_{c}=0~~\textrm{and}~~v_{c}=0.
\end{equation}
(ii) on $y=0,$
\begin{equation}\label{BC8}
u_{f} = \varphi_{f} \lambda_s \frac{\partial u_f}{\partial y}, \hspace{0.01cm} v_{f}=V_{M}(x) ,
\end{equation}
\begin{equation}\label{BC9}
u_{c} = \varphi_{c} \lambda_s \frac{\partial u_c}{\partial y} \text{     and    } v_{c}=0  .
\end{equation}
(iii) Flux condition: The non-dimensional volumetric flow rate is given by \\
\begin{equation}\label{BC10}
1=\int_{0}^{R(x)} (\varphi_{f}u_{f}+\varphi_{c}u_{c}) \hspace{0.01 cm}dy .
\end{equation}
\section{Perturbation approximation}
\noindent
In order to solve the above boundary value problem, we can use the perturbation method to find an approximate solution \citep{tsangaris1984laminar,wei2003flow,karmakar2017note}. We assume that the aspect ratio ($\delta$) of the region is small , i.e. $\delta^2 \ll 1$, and this allows us to apply perturbation theory. Accordingly, with respect to the small parameter $\delta^2$, we can expand the velocity and pressure in a perturbation series as:
\begin{equation}\label{PEq1}
(u_i, v_i, p)= (u_{i_0}, v_{i_0}, p_{0}) + \delta^2 (u_{i_1}, v_{i_1}, p_1) + \text{O}(\delta^4), \hspace{0.9cm}   \text{for } i=f, c.
\end{equation}
The first order correction is $\delta^2$, since no terms of order $\delta$ appear in the governing equations and boundary conditions. The flow field is solved by collecting the similar powers of $\delta^2$.
\subsection{The leading-order problem}
The governing equations reduce to
\begin{equation}\label{Eq1}
\frac{\partial^2 u_{f_0}}{\partial y^2} - \alpha^2\left[\left(u_{f_0} - u_{c_0}\right)+\varphi_f\frac{\partial p_0}{\partial x}\right] = 0,
\end{equation}
\begin{equation}\label{Eq3}
\frac{\partial^2 u_{c_0}}{\partial y^2} + \alpha^2\mu_r\left[\left(u_{f_0} - u_{c_0}\right)-\varphi_c\frac{\partial p_0}{\partial x}\right] = 0,
\end{equation}
\begin{equation}\label{Eq4}
\frac{\partial p_{0}}{\partial y} = 0,
\end{equation}
\begin{equation}\label{Eq5}
\left(\frac{\partial u_{f_0}}{\partial x} + \frac{\partial v_{f_0}}{\partial y}\right)= F(x,y),
\end{equation}
\begin{equation}\label{Eq6}
\left(\frac{\partial u_{c_0}}{\partial x} + \frac{\partial v_{c_0}}{\partial y}\right)= 0.
\end{equation}\\
The corresponding boundary conditions are\\\\
(i)  on $y=R(x),$\\
\begin{equation}\label{BC11}
 u_{f_0}=0, v_{f_0}=V_D(x),
 \end{equation}
\begin{equation}\label{BC12}
 u_{c_0}=0 \text{     and    } v_{c_0}=0  .
\end{equation}
(ii) on $y=0$\\
\begin{equation}\label{BC13}
u_{f_0}= \lambda_s \frac{\partial u_{f_0}}{\partial y}, v_{f_0}=V_M(x),
\end{equation}
\begin{equation}\label{BC14}
u_{c_0}= \lambda_s \frac{\partial u_{c_0}}{\partial y}\text{     and    } v_{c_0}=0  .
\end{equation}
(iii) Flux condition
\begin{equation}\label{BC15}
1=\int_{0}^{R(x)} (\varphi_{f}u_{f_0} + \varphi_{c} u_{c_0}) \hspace{0.01 cm}dy .
\end{equation}
\begin{theorem}
The vertical permeation velocities $V_D(x)$ and $V_M(x)$ are equal when
$$1=\int_{0}^{R(x)} (\varphi_{f}u_{f}+\varphi_{c}u_{c}) \hspace{0.01 cm}dy$$.
\end{theorem}
\begin{proof}
If we differentiate
$$\int_{0}^{R(x)} (\varphi_{f}u_{f_0} + \varphi_{c} u_{c_0}) \hspace{0.01 cm}dy=1,$$
under the integration sign using Leibnitz rule, we obtain
$$(\varphi_{f}u_{f_0} + \varphi_{c} u_{c_0})(x,R(x)) \frac{dR(x)}{dx} + \int_{0}^{R(x)} \left(\varphi_{f} \frac{\partial u_{f_0}}{\partial x} + \varphi_{c} \frac{\partial u_{c_0}}{\partial x}\right) \hspace{0.01 cm}dy = 0,$$
which implies
$$\left(\varphi_{f}u_{f_0}(x,R(x)) + \varphi_{c} u_{c_0}(x,R(x))\right) \frac{dR(x)}{dx} + \int_{0}^{R(x)} \left(\varphi_{f} (-\frac{\partial v_{f_0}}{\partial y}) + \varphi_{c} \left(-\frac{\partial v_{c0}}{\partial y}\right)\right) \hspace{0.01 cm}dy = 0.$$
After some simplification,
$$(\varphi_{f}v_{f_0} + \varphi_{c} v_{c_0})(y=R(x)) - (\varphi_{f}v_{f_0} + \varphi_{c} v_{c_0})(y=0) = 0.$$
Hence, $$V_D(x)=V_M(x).$$
\end{proof}
\subsection{The \text{O}($\delta^2$) problem}
The governing equations corresponding to the first order are
\begin{equation}\label{Eq7}
-\varphi_{f} \alpha^{2}\frac{\partial p_1}{\partial x} +\left(1+\frac{\lambda_{f}}{\mu_{f}}\right) \frac{\partial}{\partial x}\left(\frac{\partial u_{f_0}}{\partial x}+\frac{\partial v_{f_0}}{\partial y}\right) + \left(\frac{\partial^2 u_{f_0}}{\partial x^2} + \frac{\partial^2 u_{f_1}}{\partial y^2}\right) - \alpha^2(u_{f_1}-u_{c_1})=0,
\end{equation}
\begin{equation}\label{Eq8}
-\varphi_{f} \alpha^{2}\frac{\partial p_1}{\partial y}+\left(1+\frac{\lambda_{f}}{\mu_{f}}\right) \frac{\partial}{\partial y}\left(\frac{\partial u_{f_0}}{\partial x}+\frac{\partial v_{f_0}}{\partial y}\right) + \frac{\partial^2 v_{f_0}}{\partial y^2} -  \alpha^2 \lambda^2(v_{f_0}-v_{c_0}) = 0,
\end{equation}
\begin{equation}\label{Eq9}
-\varphi_{c} \mu_{r} \alpha^{2}\frac{\partial p_1}{\partial x} + \left(\frac{\partial^2 u_{c_0}}{\partial x^2} + \frac{\partial^2 u_{c_1}}{\partial y^2}\right) + \mu_{r}\alpha^2(u_{f_1}-u_{c_1})=0,
\end{equation}
\begin{equation}\label{Eq10}
-\varphi_{c} \mu_{r} \alpha^{2}\frac{\partial p_1}{\partial y} + \frac{\partial^2 v_{c_0}}{\partial y^2}  + \mu_{r}\alpha^2\lambda^2(v_{f_0}-v_{c_0})=0,
\end{equation}
\begin{equation}\label{Eq11}
\varphi_{f}\left(\frac{\partial u_{f_1}}{\partial x} + \frac{\partial v_{f_1}}{\partial y} \right) = 0,
\end{equation}
\begin{equation}\label{Eq12}
\varphi_{c}\left(\frac{\partial u_{c_1}}{\partial x} + \frac{\partial v_{c_1}}{\partial y} \right) = 0.
\end{equation}\\
Corresponding boundary conditions reduce to\\\\
(i)  on $y=R(x),$\\
\begin{equation}\label{BC16}
 u_{f_1}=2\pi a^{2}\sin(2\pi x)V_{D}(x),\hspace{0.2cm} v_{f_1}=-2\pi^{2} a^{2}\sin(2\pi x)V_{D}(x),
\end{equation}
\begin{equation}\label{BC17}
u_{c_1}=0 \text{     and    } v_{c_1}=0  .
\end{equation}
(ii) on $y=0,$\\
\begin{equation}\label{BC18}
u_{f_1}=  \lambda_s \frac{\partial u_{f_1}}{\partial y},\hspace{0.2cm} v_{f_1}=0,
\end{equation}
\begin{equation}\label{BC19}
u_{c_1}= \lambda_s \frac{\partial u_{c_1}}{\partial y} \text{     and    } v_{c_1}=0  .
\end{equation}
(iii) Flux condition \\
\begin{equation}\label{BC20}
0=\int_{0}^{R(x)} (\varphi_{f}u_{f_1} + \varphi_{c} u_{c_1}) \hspace{0.01 cm}dy.
\end{equation}
\subsection{Structure of the source term $F(x,y)$}
\noindent
In Eq. (\ref{Eq5}), $F(x,y)$ represents the tip of the needle of the syringe at some point inside the SCL. One can think of a point source at the point $(0,y_{0})$ (see Fig. \ref{Geometry}) which can be expressed as
\begin{equation}\label{Eq13}
F(x,y)=m_{0}\delta(x)\delta(y-y_{0}),
\end{equation}
where $m_{0}$ represents the strength of the point source. One can solve the leading order and first order equations using finite difference scheme by discretizing the domain while keeping the point $(0,y_{0})$ outside the meshgrid. However, one can attempt for the analytical solution in the regions $y<y_0$ and $y>y_0$ for all values of $x$. Since $y=y_{0}$ is a line on which injection point (tip of needle) must lie, thus the following conditions at $y=y_0$ can be used to match the solution :
$$u_{f_0}(x,y^+_{0}) = u_{f_0}(x,y^-_{0}),$$
$$u_{c_0}(x,y^+_{0}) = u_{c_0}(x,y^-_{0}),$$
and\\
$$\frac{\partial u_{f_0}(x,y)}{\partial y}\Big|_{y=y^+_0} = \frac{\partial u_{f_0}(x,y)}{\partial y}\Big|_{y=y^-_0}  ,$$
$$\frac{\partial u_{c_0}(x,y)}{\partial y}\Big|_{y=y^+_0} = \frac{\partial u_{c_0}(x,y)}{\partial y}\Big|_{y=y^-_0}  .$$
\subsection{Subcutaneous tissue velocity and stream function}
We define subcutaneous tissue velocity or composite velocity $\mathbf{u}=(u,v)$ of the mixture of IF and AC presents in the SCL as
\begin{equation}\label{Eq14}
  \mathbf{u} = \varphi_f\mathbf{u}_{f} + \varphi_c\mathbf{u}_{c}
\end{equation}
If $\Phi_f$ and $\Phi_c$ are the stream function of the IF and AC respectively that are present in the SCL, then their relation with the velocity components are
 $$u_{f}=\frac{\partial \Phi_{f}}{\partial y}, u_{c}=\frac{\partial \Phi_{c}}{\partial y}  \text{       and      } v_{f} = - \frac{\partial \Phi_{f}}{\partial x}, v_{c} = - \frac{\partial \Phi_{c}}{\partial x} .$$
 Also, we define the composite stream function of the mixture as
\begin{equation}\label{Eq15}
  \Phi = \varphi_f\Phi_{f} + \varphi_c\Phi_{c}
\end{equation}
In the SCL, the quantity of IF is much larger than the quantity of AC. Thus the subcutaneous tissue velocity or composite velocity and stream function can be considered in macroscopic level \citep{barry1997deformation,barry1991fluid}. The spreading of the injected fluid within the interstitial space of the subcutaneous tissue can be manifested by the streamline pattern exhibited by the composite motion of IF and AC.
\begin{theorem}
If $\Phi_i (i=f,c)$ are the stream functions defined as
  \[\Phi_i = \begin{cases}
      \Phi_i^{(1)}, & 0\leq y<y_0 \\
      \Phi_i^{(2)}, & y_0 < y \leq R(x)
   \end{cases}
\]
and there exists $\psi_{i}$ satisfying $\psi_{i}=\int u_{i} \hspace{0.03 cm} dy$ such that
\[\psi_i = \begin{cases}
      \psi_i^{(1)}, & 0\leq y<y_0 \\
      \psi_i^{(2)}, & y_0 < y \leq R(x),
   \end{cases}
\]
then the following relations hold
$$\Phi_i^{(1)}(x,y) = \psi_{i}^{(1)}(x,y) + f_{i}(x) ,$$
$$\Phi_i^{(2)}(x,y) = \psi_{i}^{(2)}(x,y) + g_{i}(x) ,$$
for arbitrary function $f_{i}(x)$ and $g_{i}(x)$.
\end{theorem}
\begin{proof}
The stream function is related to the velocity components by the relations
$$u_{i}=\frac{\partial \Phi_{i}}{\partial y} \text{       and      } v_{i} = - \frac{\partial \Phi_{i}}{\partial x} .$$
Considering the first relation $$u_{i}=\frac{\partial \Phi_{i}}{\partial y} ,$$
or $$\Phi_{i} = \int u_{i}\hspace{0.03 cm} dy + h(x) ,$$
where $h(x)$ is obtained from the integration. Let $\psi_{i}(x,y)=\int u_{i}\hspace{0.03 cm} dy$.

\noindent
Since $u_{i}$ and $v_{i}$ are defined in two regions (say $u_{i}^{(1)}$, $u_{i}^{(2)}$  and $v_{i}^{(1)}$, $v_{i}^{(2)}$), thus we have stream function in the regions as
$$\Phi_i^{(1)}(x,y) = \psi_{i}^{(1)}(x,y) + f_{i}(x), \hspace{1 cm} \text{if} \hspace{0.3 cm} 0\leq y<y_0 $$
$$\Phi_i^{(2)}(x,y) = \psi_{i}^{(2)}(x,y) + g_{i}(x). \hspace{1 cm}~~~~~~ \text{if} \hspace{0.3 cm} y_0\leq y<R(x) $$
Next we have to find the functions $f_{i}(x)$ and $g_{i}(x)$.\\

\noindent
Since stream function is continuous, thus we have
$$\Phi_i^{(1)}(x,y_0^{-})=\Phi_i^{(2)}(x,y_0^{+}) ,$$
which gives $$f_{i}(x) = \psi_{i}^{(2)}(x,y) - \psi_{i}^{(1)}(x,y) + g_{i}(x) .$$
Also since $$\Phi_i^{(2)}(x,y) = \psi_{i}^{(2)}(x,y) + g_{i}(x) ,$$
Upon taking derivative on both sides with respect to $x$, we have
$$ \frac{\partial \Phi_i^{(2)}}{\partial x} = \frac{\partial \psi_{i}^{(2)}}{\partial x} + g_{i}'(x) ,$$
which gives
$$ g_{i}'(x) = - \frac{\partial \psi_{i}^{(2)}}{\partial x} - v_{i}^{(2)}(x,y) .$$
Now integrating both sides, we get
$$g_{i}(x)=-\psi_{2}(x,y) - \int v_{i}^{(2)}(x,y) \hspace{0.03 cm} dx ,$$
which satisfies throughout the considered region.\\
Thus at $y=R(x)$, \\
$$g_{i}(x) = -\psi_{2}(x,R(x)) - \int v_{i}^{(2)}(x,R(x)) \hspace{0.03 cm} dx .$$
Since $v_{i}^{(2)}(x,R(x))$ is a known function (which we obtain using boundary conditions), thus we get $g_{i}(x)$ and using this we can obtain $f_{i}(x)$ easily from the continuity condition of stream function.
\end{proof}
\noindent
The detailed solution of the leading order and \text{O}($\delta^2$) problem corresponding to IF and AC are shown in Appendix A and Appendix B respectively.
\section{Results and Discussion}
\noindent
In this study, a flow-induced by fluid injection has been considered within the anisotropic SCL which is bounded by permeable DL from the topside and permeable ML from the bottom. As per the present mathematical model is concerned, the principal components of the SCL are AC (fat tissue) and IF with a large proportion of fluid part. Consequently, $\varphi_f$ is chosen within the range $0.7\leq \varphi_f \leq 0.9$ throughout the study (see \citet{khor1991potential,truskey2010transport}). All the flow parameters such as $\lambda$, $\mu_r$, $\delta$, $a$, $\text{Da}$, $\lambda_s$ are reported in Table \ref{table_1} with their reference ranges and are chosen based on experimental or theoretical studies already existed in literature. The analysis underlying the consideration of parameter ranges is discussed below. For example, $(Da)$ which is the ease of fluid percolation in the horizontal direction can be considered within the range $10^{-3}\leq Da\leq 5\times10^{-3}$ following \citet{dey2016hydrodynamics}. Based on the choice of $Da$, we specify the slip coefficient $\lambda_s$ within the range $0.001\leq \lambda_{s} \leq 0.05$ as the value of it is up to $\text{O}(\sqrt{\text{Da}})$.
\begin{table}[h!]
\centering
\begin{tabular}{p{42mm} p{40mm} p{40mm}}
\hline
 Parameter & Range & Remark \\
 \hline
 Anisotropic ratio ($\lambda$) & $0.5<\lambda\leq2$ & \citet{shrestha2020imaging}\\
 Viscosity ratio ($\mu_r$) & $0<\mu_r<1$ & Considered\\
$Da$ & $10^{-3}\leq Da\leq 5\times10^{-3}$ & \citet{dey2016hydrodynamics}\\
 Amplitude of the wavy layer ($a$) & $0.34\leq a\leq 0.5$ & Considered \\
 Aspect ratio of the region ($\delta$) & $0.1\leq\delta\leq 0.3$ & \citet{karmakar2017note}\\
 Slip coefficient ($\lambda_s$) & $0.001\leq\lambda_s\leq0.05$ & \citet{dey2016hydrodynamics}\\
 \hline
\end{tabular}
\caption{Various parameters involved in this study with their range.}
\label{table_1}
\end{table}
Note that $\lambda_{s}\rightarrow 0$ makes the adipose cellular phase rigid towards the squeezing effect at the SM interface due to the fluid pressure at the line of injection. On the other hand, the choice of $\lambda$ is prompted by the study of \citet{shrestha2020imaging} which reports the optimum value of permeability of skin tissue lies within the range $0.59\times 10^{-14}$ $\textrm{m}^{2}$ to $2.10\times 10^{-12}$ $\textrm{m}^{2}$. Therefore, one can consider $K_{1}$ and $K_{2}$ lying between the above range. Consequently, $\lambda$ lies within the range $0.53\leq \lambda \leq 2$. Except for $\lambda=1$ (isotropic), anisotropic behavior is exhibited for all values of $\lambda$ within the above range. The perturbation parameter $\delta$ is the ratio between the SCL depth and the SCL length trapped inside the thumb and pointer of the medical staff during the injection. Essentially, we have $\delta^2\ll1$ for this study. A similar parameter has been observed in the study of \citet{karmakar2017note} and consequently, we opt the magnitude of $\delta$ within the range $0.1\leq\delta\leq 0.3$. A sensitivity analysis associated with the choosing $\delta$ with the help of $\lambda$ is shown in Table \ref{table_3}. In general, the adipose cellular phase should have a higher viscosity compared to the interstitial fluid (IF) within the same continuum description. This is because adipose cells, or adipocytes, are lipid-rich cells with a lipid viscosity of $36.8$ mPa.s or $36.8\times 10^{-3}$ $\textrm{kg} \textrm{m}^{-1} \textrm{s}^{-1}$ \mbox{\citep{comley2010micromechanical}}, while the viscosity of the interstitial fluid can be approximated as $3.5 \times 10^{-3}$ $\textrm{kg} \textrm{m}^{-1} \textrm{s}^{-1}$ by following the study of \mbox{\citet{yao2012interstitial}}. Hence, the viscosity ratio $\mu_r=\mu_f/\mu_c$ can be chosen to lie within the range $0<\mu_r<1$.\\
\begin{table}[h!]
\centering
\begin{tabular}{p{20mm} p{20mm} p{20mm}}
 \hline
 $\lambda$ & $\delta^2(\ll1)$ & $\delta^2\lambda^2(\ll1)$ \\
 \hline
 1 & 0.09 & 0.09 \\
 1.5 & 0.09 & 0.2025\\
 1.75 & 0.09 & 0.2756\\
 2 & 0.09 & 0.36\\
 \hline
\end{tabular}
\caption{Various tissue anisotropic ratio ($\lambda$) magnitude, aspect ratio ($\delta$) value and corresponding value of $\delta^2\lambda^2$.}
\label{table_3}
\end{table}
\noindent
$V_{M}(x)$ represents the permeation velocity at the SM interface which can be expressed as follows
\begin{equation}\label{oth1}
V_{M}(x) = V_{M_0} + \sigma_{M} \left(\frac{d p_{0}}{d x}\right),
\end{equation}
where $V_{M_0}$ is the permeation velocity at constant leading order pressure. Note that Eq. (\mbox{\ref{oth1}}) is analogous to Darcy law defined at the SM interface. Understanding the impact of the coefficient $\sigma_{M}$ on permeation regulation is absolutely crucial, particularly in the presence of varying pressure. It should be noted that $V_{M}(x)$ is directly linked to the gradient of the leading order pressure, as it represents the vertical component of the interfacial velocity at $y=0$. Evidently, the minimum value of $V_{M}(x)$ is $V_{M_{0}}$, assuming a positive pressure gradient near the SM interface. In order to determine the minimum value of $V_{M}(x)$, it is imperative that we ensure ${d p_{0}}/{d x}$ is equal to zero on SM interface (please refer Fig. \mbox{\ref{Geometry}}). At the point $(x,y)=(0,0)$ on the SM interface, the skin pinching depth attains its maximum as per the schematic in Fig. \mbox{\ref{Geometry}}. Some algebra on ${d p_{0}}/{d x}$ obtains
\begin{equation}\label{oth1-2}
\frac{d p_{0}}{d x} = \frac{\alpha^2}{4} \mu_r(1+\mu_r) \left(a-\frac{1}{3}\right),~~~~\textrm{at}~~~~x=0.
\end{equation}
Clearly, $\left({d p_{0}}/{d x}\right)_{x=0}\geq 0$ as $a\geq{1}/{3}$. Hence, $V_{M}(x=0)$ is the permeation velocity at $x=0$ from SCL to ML. Permeation velocity can also be achieved at other points $x=x_{0}$ on the SM interface, where $x_{0}$ is not equal to zero. $\sigma_{M}$, as per Eq. \mbox{(\ref{oth1})}, represents the difference in permeation velocity at $x=o$ and the lowest possible permeation velocity per unit gradient of pressure calculated at $x=0$ on the SM interface. Therefore, from Eqs. (\mbox{\ref{oth1}})-(\mbox{\ref{oth1-2}}) it follows
\begin{equation}\label{oth2}
\frac{V_{M}(0) - V_{M_0}}{\sigma_{M}}=\frac{\alpha^2}{4}\mu_r(1+\mu_r)\left(a-\frac{1}{3}\right),
\end{equation}
which says $V_{M}(0)$ is equal to $V_{M_0}$ when the value of $a$ is equal to $1/3$. This means that more than one-third of the SCL depth needs to be pinched up. Eventually, Table \mbox{\ref{table_2}} indicates that the maximum depth of SCL produced for injection varies depending on the type of skin lifting. Fig. \mbox{\ref{vmx_mur}} demonstrates a noticeable increase in the difference between $V_{M}(0)$  and $V_{M_0}$ as the value of $\mu_{r}$ increases for various pinching depths between $a=0.35$ and $a=0.4$. This suggests that the permeation velocity increases with the pinching depth and $\mu_{r}$. To achieve optimal permeation, increasing the skin pinching depth beyond $a=1/3$ is recommended. Similar analysis as above can be done at other locations of SM interface. Additionally, using highly viscous fluid may enhance permeation, although the negative effects of high viscosity must be considered. Further discussion on this topic can be found in the study later.\\
\begin{figure}[h!]
\centering
\includegraphics[width=0.5\linewidth]{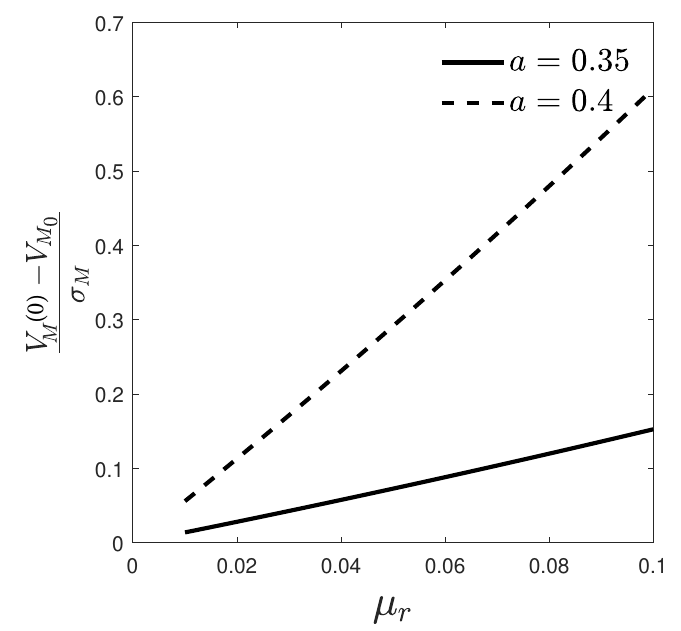}
\caption{Variation of $\left(V_{M}(0) - V_{M_0}\right)/\sigma_{M}$ with respect to $\mu_{r}$ when $a=0.35$ and $a=0.4$ corresponding to $Da = 3\times 10^{-3}$.}
\label{vmx_mur}
\end{figure}
\begin{table}[h!]
\centering
\begin{tabular}{p{85mm} p{42mm}}
 \hline
 Nature of skin lifting & Maximum depth of the SCL for injection \\
 \hline
 Higher pinching depth ($a=0.45$) & $R(0)=1.45$  \\
 Moderate pinching depth ($a=0.4$) & $R(0)=1.4$ \\
 Minimum required pinching depth ($a=0.34$) & $R(0)=1.34$ \\
 \hline
\end{tabular}
\caption{Various types of skin lifting and the corresponding height of the skin produced for Subcutaneous injection}
\label{table_2}
\end{table}
\noindent
In order to understand the flow pattern of injected fluid within the SCL, we compute the streamlines corresponding to the composite velocity. The corresponding flow patterns are recognized and explored using three important parameters $a$, $\mu_{r}$ and $\lambda$. The behavior of the flow patterns are discussed using axial composite velocity as a function of $y$ and computed shear stress at three positions of $y$: (i) line of injection $(y=y_{0})$ (ii) SM interface $(y=0)$ and (iii) SD interface $(y=R(x))$. In the upcoming sections, we are going to discuss this in detail.
\subsection{Flow pattern of the injected fluid in terms of Composite Streamlines}
\noindent
The injection fluid's flow patterns containing the drug are illustrated in Figs. \ref{contour_a3}-\ref{contour_a4} which display composite streamlines ($\Phi$) for two different values of $a=0.34$ and $a=0.4$ respectively. The remaining parameters assume the following values: $\text{Da}=3\times10^{-3}$, $\delta=0.3$, $\lambda = 2$, $\mu_r = 0.01$, $\lambda_{s} = 0.05$. The development of closed contours is observed from the line of injection in the form of primary eddy due to the evolution of high pressure. In addition, the streamlines follow the curvature of the SCL and generate additional closed contours as secondary eddy structures within the lifted portion of the SCL. The secondary eddy structure becomes more pronounced with higher pinching depth $a$. If we locate a particular contour, say $c=0.8$, formed near the line of injection, the corresponding size increases with $a$. This phenomenon is due to the constant transfer of kinetic energy (KE) from a large to a small eddy until the dissipation of KE. The formation of eddies helps better mixing of drug loaded injected fluid with the IF than pure molecular diffusion. In other words, a higher pinching depth causes better assimilation of the injected fluid with the IF.\\
\begin{figure}[h!]
\centering
\subfigure[]{\label{contour_a3}\includegraphics[width=0.38\textwidth]{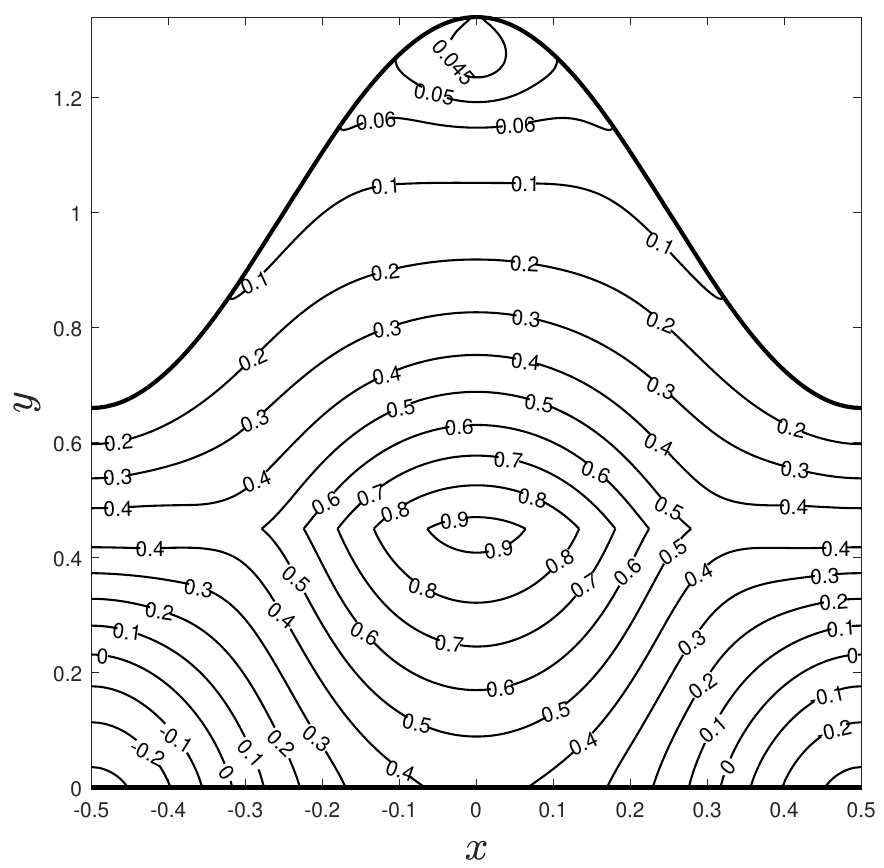}}
\subfigure[]{\label{contour_a4}\includegraphics[width=0.38\textwidth]{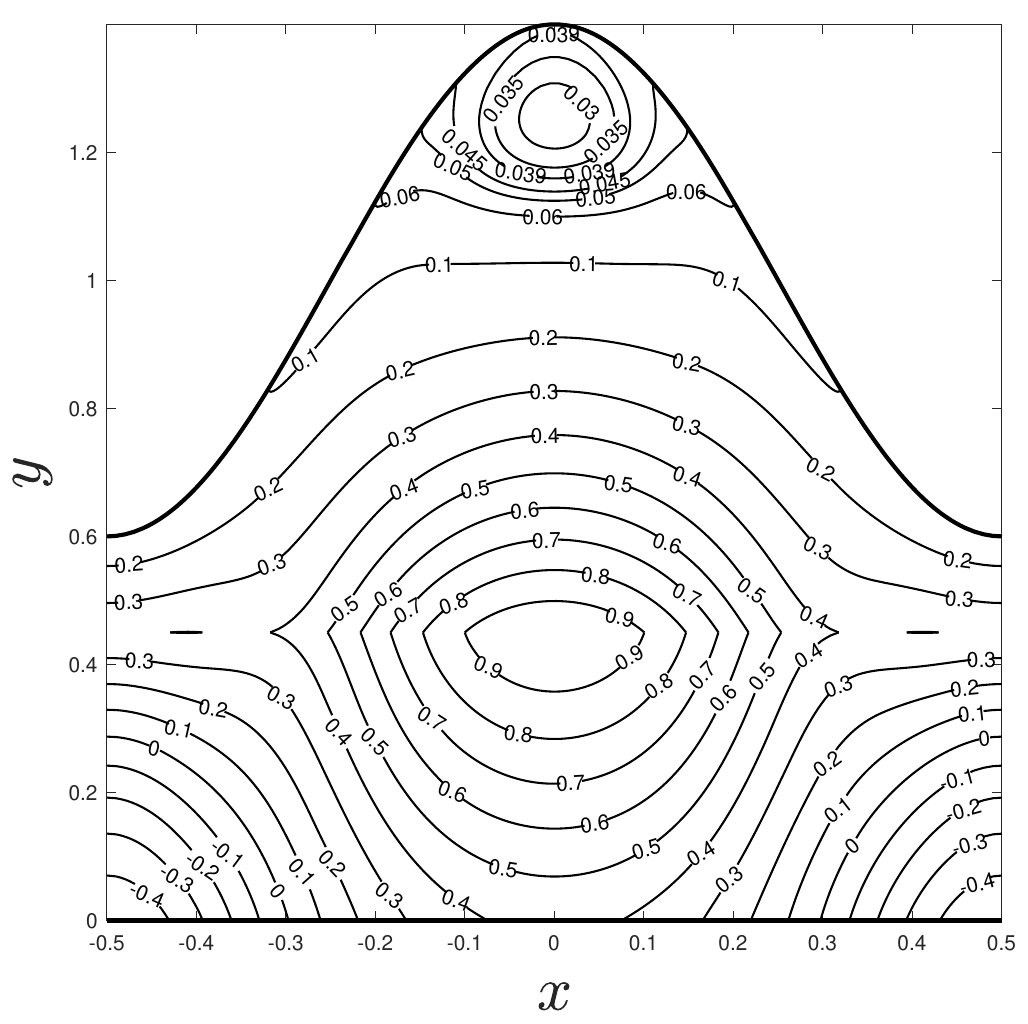}}
\caption{Composite streamlines within the subcutaneous layer for two different pinching depths (a) $a=0.34$, (b) $a=0.4$ with other parameters are $\lambda=2$, $Da=3\times10^{-3}$, $\delta=0.3$, $\mu_r = 0.01$ and $\lambda_s=0.05$, when $y_0=0.45$ is the line of injection.}
\end{figure}

\noindent
The variations in the flow pattern of the injected fluid inside SCL are explained in terms of $\Phi$ and depicted through Figs. \ref{contour_lambda_1}-\ref{contour_lambda_2} for a wide variety of tissue anisotropy ($\lambda$). Among all considered values of $\lambda$, Fig. \ref{contour_lambda_1} corresponds to the isotropic nature of the SCL, and the rest are plotted for anisotropic nature, in particular when the horizontal permeability $K_1$ is greater than the vertical permeability $K_2$. When $\textrm{Da}$ is fixed, an increase in $\lambda$ results in a reduction of the tissue anisotropy along the vertical direction. Consequently, the streamlines in the upstream follow the shape of the SD interface where the skin is pinched up until $\lambda\thickapprox1.5$ (see Fig. \ref{contour_lambda_1}). After that, the flow in the upstream starts to digress following the shape of the interface (see Fig. \ref{contour_lambda_1.5}). Corresponding to $\lambda=1.75$ and $2$ onset of the secondary eddy structure causing flow circulation can be seen at the position where the skin is lifted (see Figs. \ref{contour_lambda_1.75} and \ref{contour_lambda_2}). As discussed in the previous paragraph, the secondary eddy structure aftermaths good mixing of injected fluid with the IF. Hence, the movement of the injected drug within the SC tissue region becomes vigorous in case of larger tissue anisotropy and high lifting of the skin.\\
\begin{figure}[h!]
\centering
\subfigure[]{\label{contour_lambda_1}\includegraphics[width=0.38\textwidth]{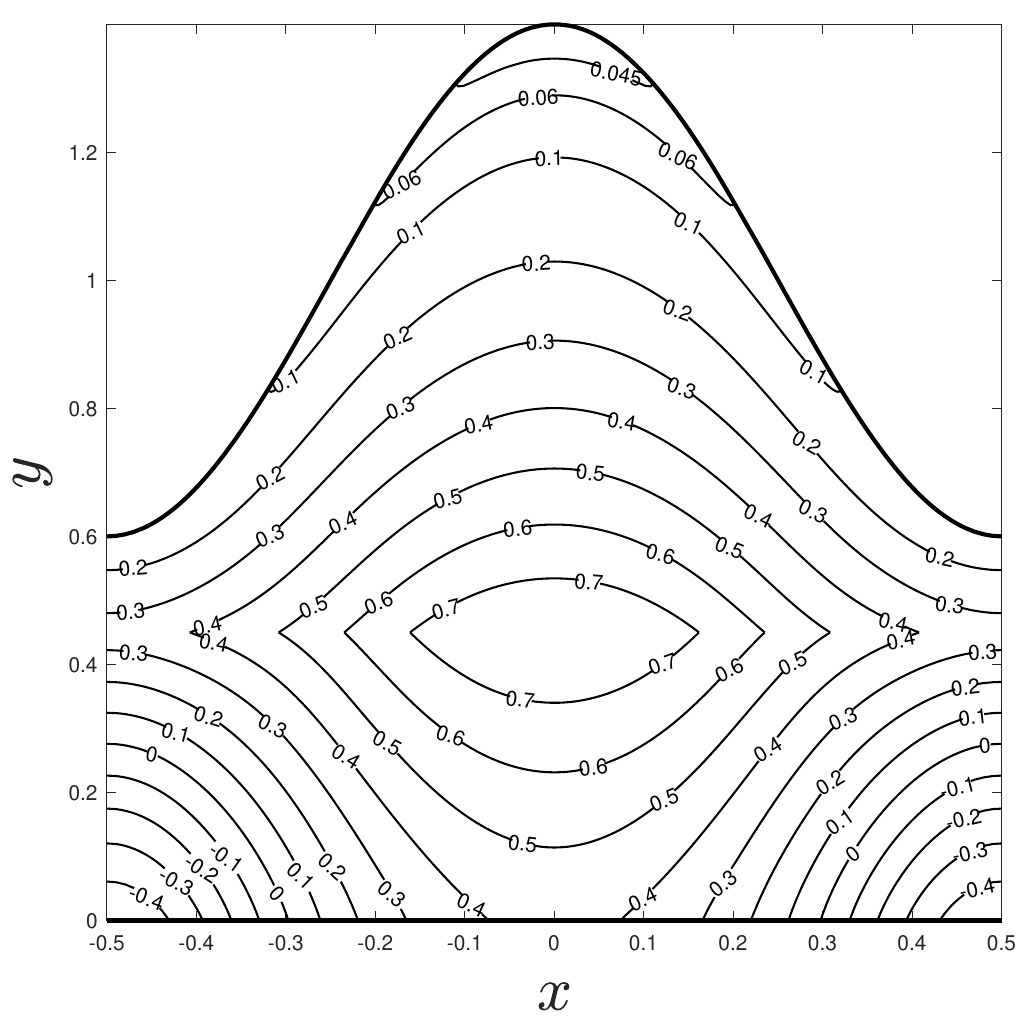}}
\subfigure[]{\label{contour_lambda_1.5}\includegraphics[width=0.38\textwidth]{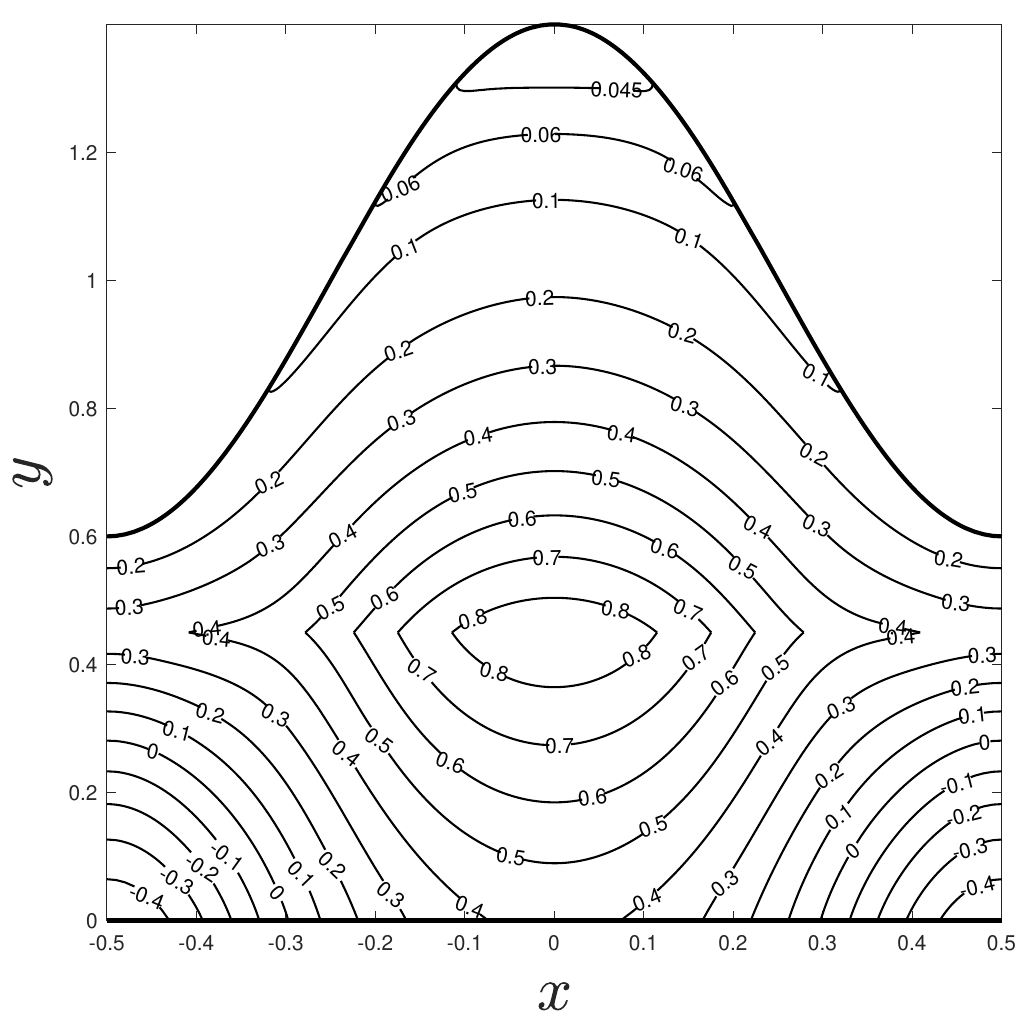}}
\subfigure[]{\label{contour_lambda_1.75}\includegraphics[width=0.38\textwidth]{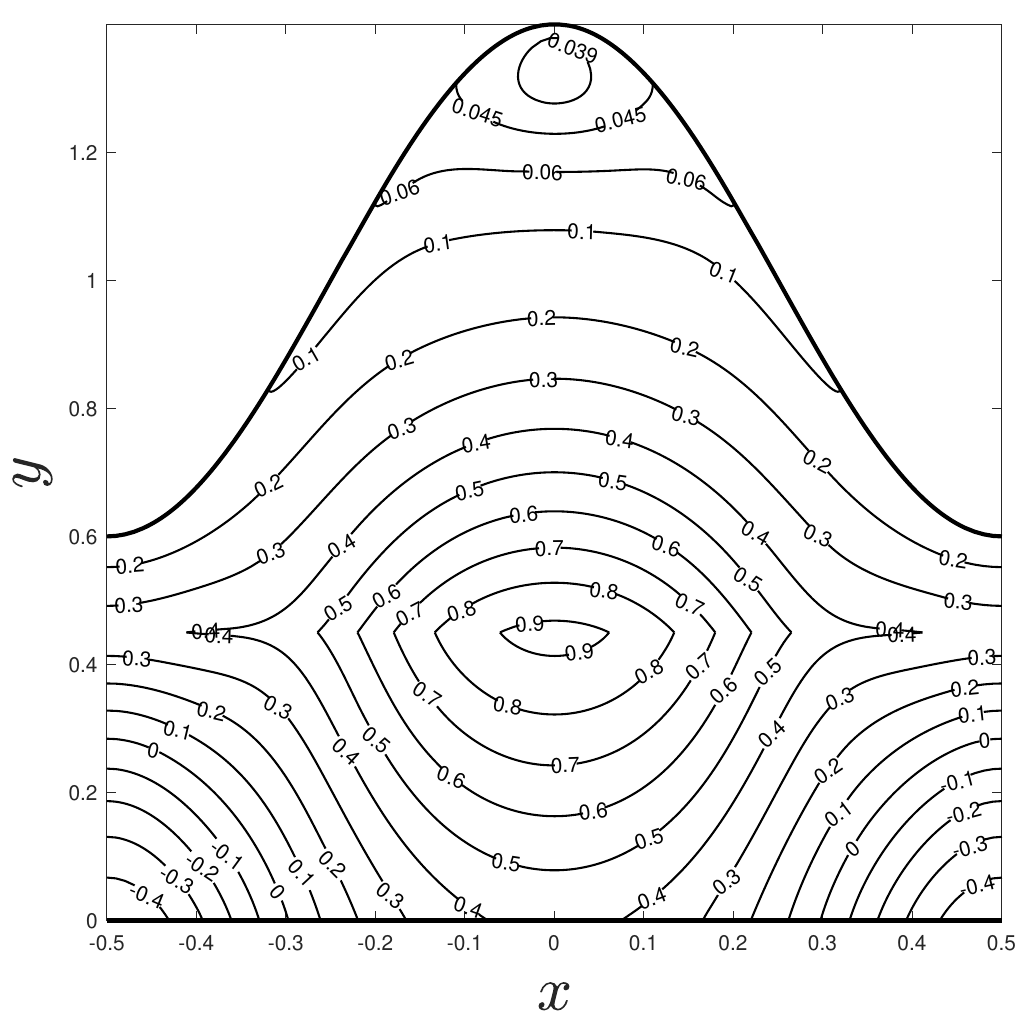}}
\subfigure[]{\label{contour_lambda_2}\includegraphics[width=0.38\textwidth]{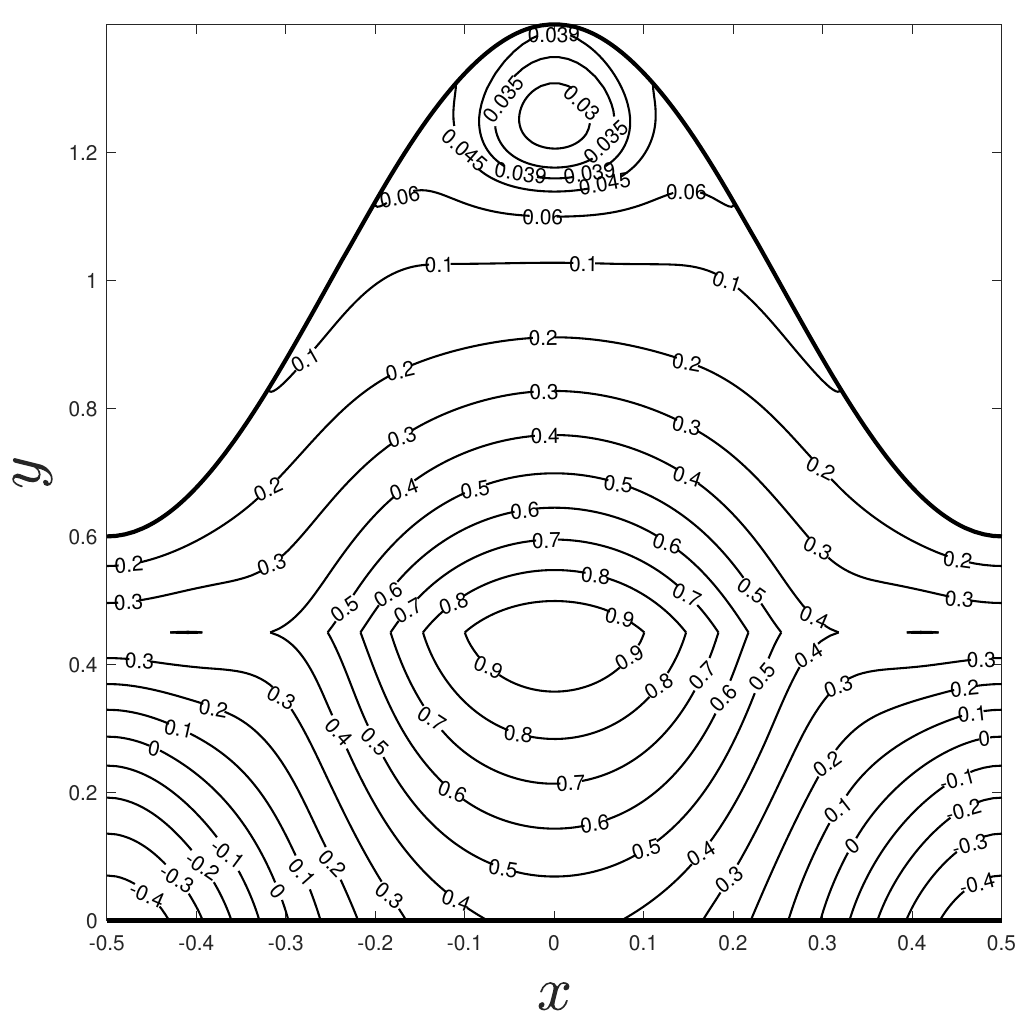}}
\caption{Composite streamlines within subcutaneous layer for different tissue anisotropy ratio (a) $\lambda=1$, (b) $\lambda=1.5$, (c) $\lambda=1.75$ and (d) $\lambda=2$ with $a=0.4$, $Da=3\times10^{-3}$, $\delta=0.3$, $\mu_r = 0.01$ and $\lambda_s=0.05$, when $y_0=0.45$ is the line of injection.}
\end{figure}

\noindent
In general, the IF is less viscous than that of the AC population. Consequently, $\mu_r$ is less than 1 (i.e., $\mu_f<\mu_c$). So, the AC phase experiences less drag from the IF side. In other words, the AC can impose high interstitial resistance towards IF movement during fluid injection. Figs. \ref{contour_mur_0.1}-\ref{contour_mur_0.01} illustrate the flow pattern of the injected fluid for $\mu_r=0.1$, $\mu_r=0.05$ and $\mu_r=0.01$. Only primary eddies are developed near the SD interface for $\mu_r=0.1$. But corresponding to the reduced $\mu_r$, a significant viscosity difference between AC and IF is developed. The development of secondary eddy can be observed at the lifted portion for $\mu_r=0.05$, which becomes prominent with a further decrement of $\mu_r$. We can locate a contour $c=0.7$ when the viscosity ratio is as low as $\mu_r=0.1$ close to the line of injection, which subsequently increases in size for smaller $\mu_r$ (Fig. \ref{contour_mur_0.05}). Hence $\lambda \geq 1.75$, $a\geq0.4$, and $\mu_r\leq0.01$ are the conditions to be satisfied simultaneously to support eddy structure within the lifted portion of the skin.
\subsection{Flow pattern of the injected fluid in terms of Composite Velocity}
\noindent
In order to justify the behavior of the composite streamlines, we go through the axial composite velocity ($u$) profiles as shown in Figs. \ref{velocity_a}-\ref{velocity_mur} for the above three parameters $a$, $\lambda$ and $\mu_{r}$. We identify three intervals for $y$ such as (i) $0 < y \leq 0.4$ (ii) $0.4 < y \leq 1$ (iii) $1\leq y \leq 1.4$ in which $u$ shows varied behavior due to the assumption of constant volumetric flux condition. This phenomenon justifies the dissipation of a particular contour within the primary eddy after getting larger (with an increase in $a$) and simultaneously creating new smaller contours within the secondary eddy. Consequently, an increased magnitude of the axial composite velocity with $a$ is noted for ranges of $y$ in (i) and (iii), indicating the role of enhanced axial convective transport corresponding to deeper skin pinching (see Fig. \ref{velocity_a}). On the other hand, within the range $0.4 < y \leq 1$, $u$ shows change in sign for both $a=0.34$ and $a=0.4$ indicating the formation of the secondary eddy.\\

\begin{figure}[h!]
\centering
\subfigure[]{\label{contour_mur_0.1}\includegraphics[width=0.38\textwidth]{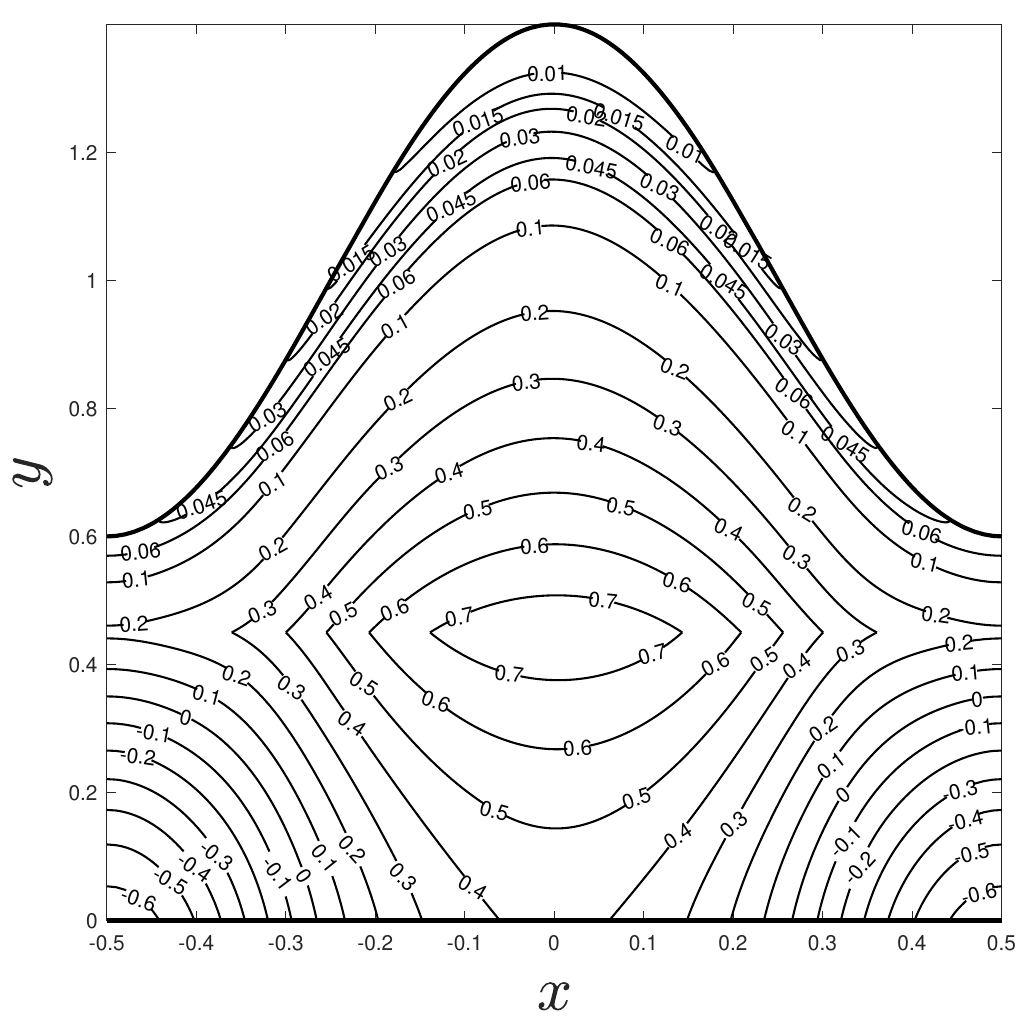}}
\subfigure[]{\label{contour_mur_0.05}\includegraphics[width=0.38\textwidth]{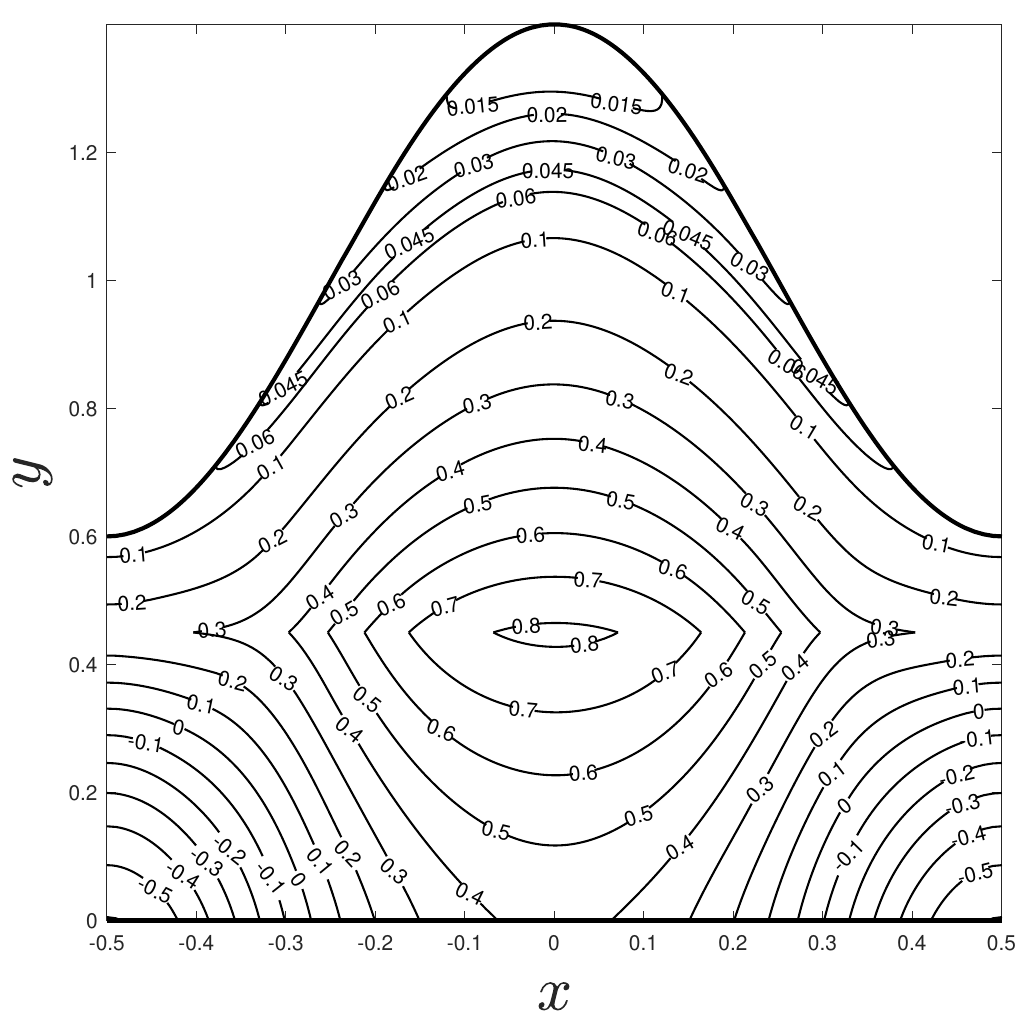}}
\subfigure[]{\label{contour_mur_0.01}\includegraphics[width=0.38\textwidth]{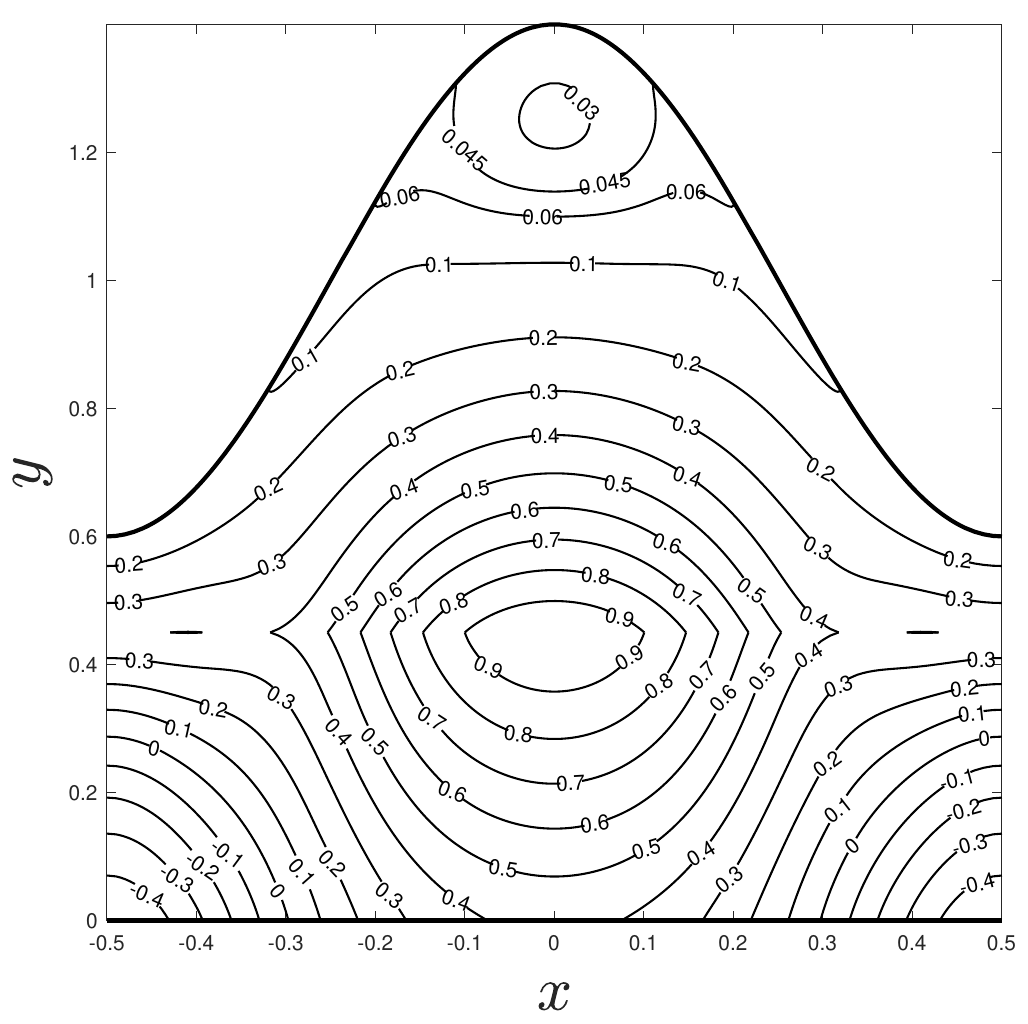}}
\caption{Patterns of composite streamlines within subcutaneous layer for different viscosity ratio (a) $\mu_r=0.1$, (b) $\mu_r=0.05$ and (c) $\mu_r=0.01$ with the other parameters are $\lambda=2$, $Da=3\times10^{-3}$, $\delta=0.3$, $a = 0.4$ and $\lambda_s=0.05$, when $y=y_0(=0.45)$ is the line of injection.}
\end{figure}
\noindent
The influence of tissue anisotropy on generating secondary eddy can be discussed using the axial composite velocity $u$. Fig. \ref{velocity_lambda} represents profiles of $u$ versus $y$ at $x=0$ for various $\lambda$ when $a=0.4$, $\mu_r = 0.01$, $\lambda_{s} = 0.05$, $\text{Da}=3\times10^{-3}$, and $\delta=0.3$. At the SD interface, $u$ becomes zero due to the no-slip condition, while the IF and AC exhibit horizontal motion at the SM interface due to the slip velocity associated with the squeezing effect under high pressure developed due to SC injection. Except in the interval $0.4\leq y \leq 1$, $u$ increases with $\lambda$ while an opposite behavior is noted within the stated region. Such contrasting behavior of $u$ is due to the fixed volumetric flow rate across the SC region. Moreover, $u$ does not change its sign for $\lambda=1, 1.5$. Still, it changes from positive to negative for $\lambda=1.75, 2$ indicating the development of a secondary eddy structure near the SD interface where the skin is lifted.\\

\noindent
Finally, Fig. \ref{velocity_mur} shows that the composite velocity near the SD and SM interface is decreased with an increasing magnitude of $\mu_r$. However, the opposite phenomenon is noticed near the line of injection. Also, axial composite velocity changes its sign for $\mu_r=0.01$ while it does not change for $\mu_r=0.05$ and $\mu_r=0.1$ which becomes the root cause of prominent secondary eddy structure corresponding to $\mu_r=0.01$. The development of two eddies may be discussed with the help of pressure gradient and shear stress at the three positions of $ y$, i.e., $y=0$, $y=y_0$, and $y=R(x)$.
\begin{figure}[h!]
\centering
\subfigure[]{\label{velocity_a}\includegraphics[width=0.38\textwidth]{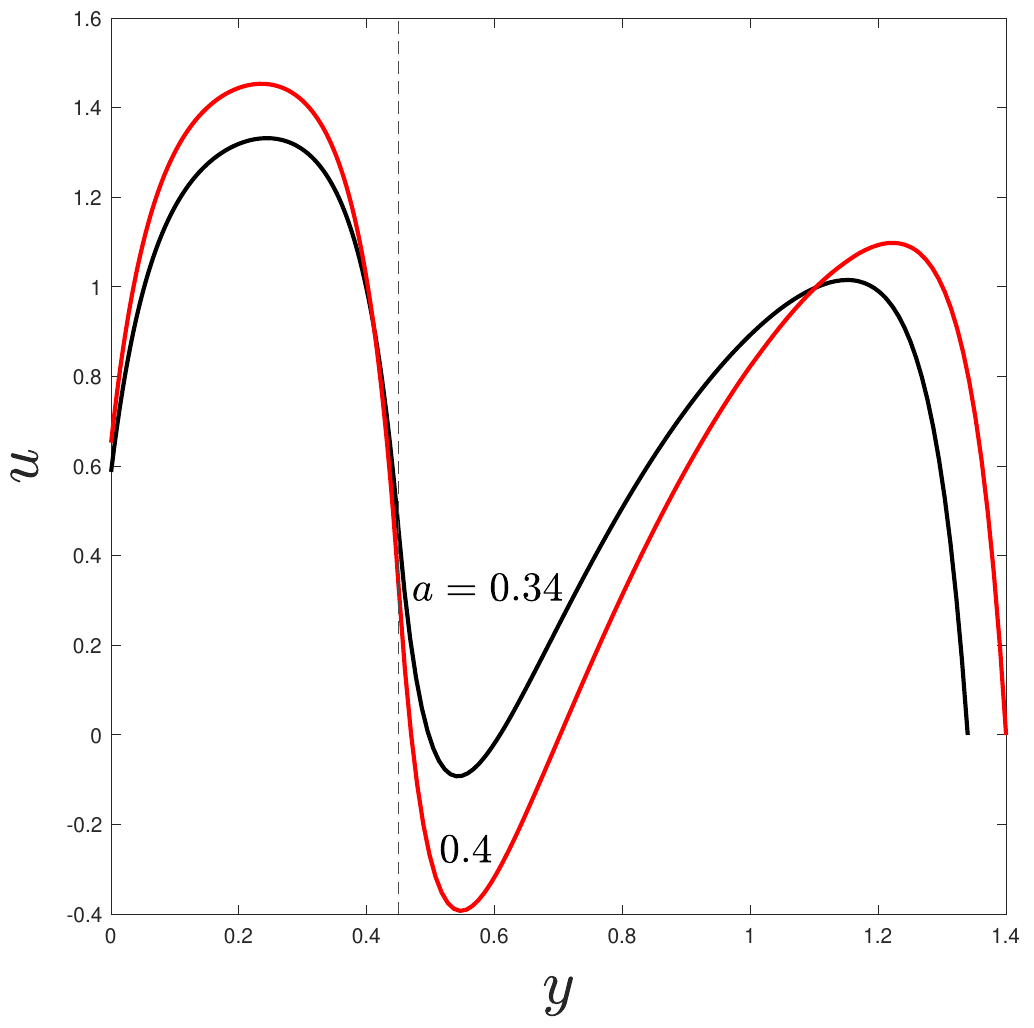}}
\subfigure[]{\label{velocity_lambda}\includegraphics[width=0.38\textwidth]{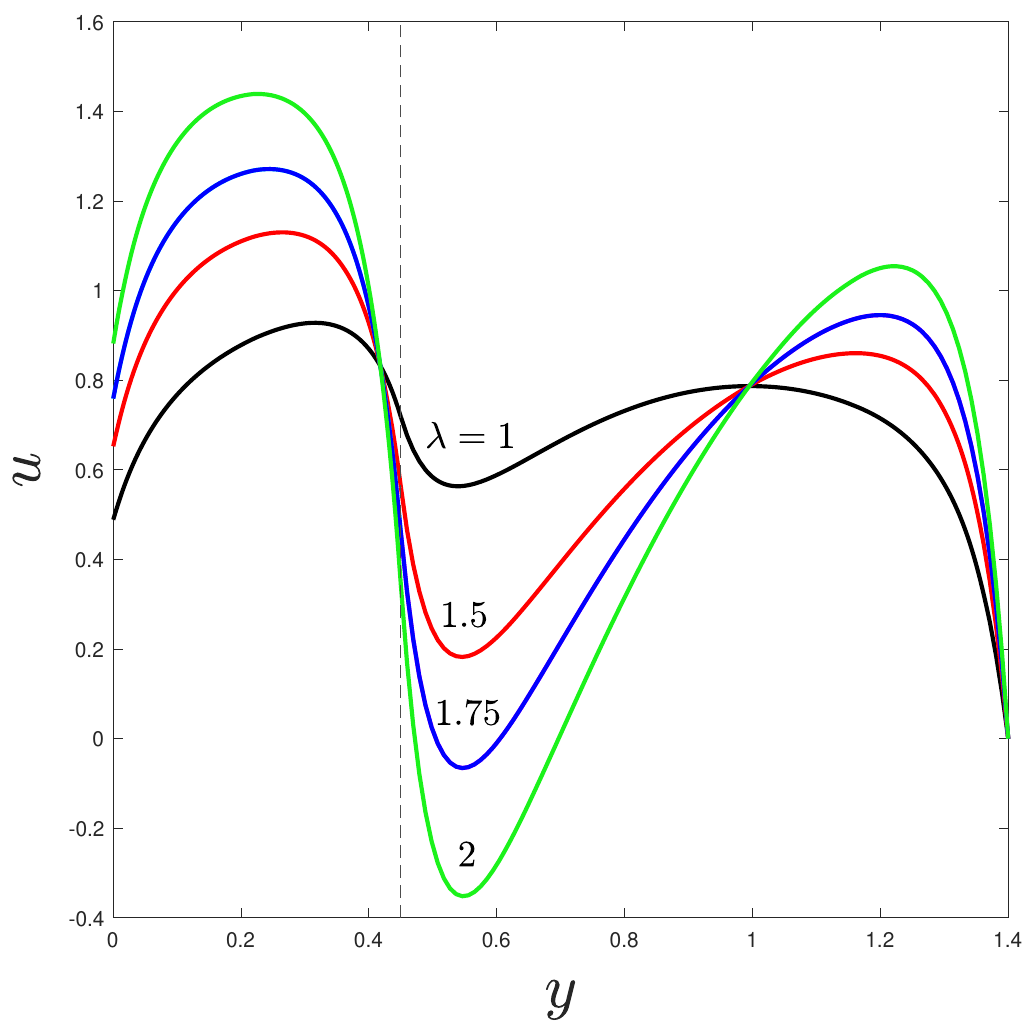}}
\subfigure[]{\label{velocity_mur}\includegraphics[width=0.38\textwidth]{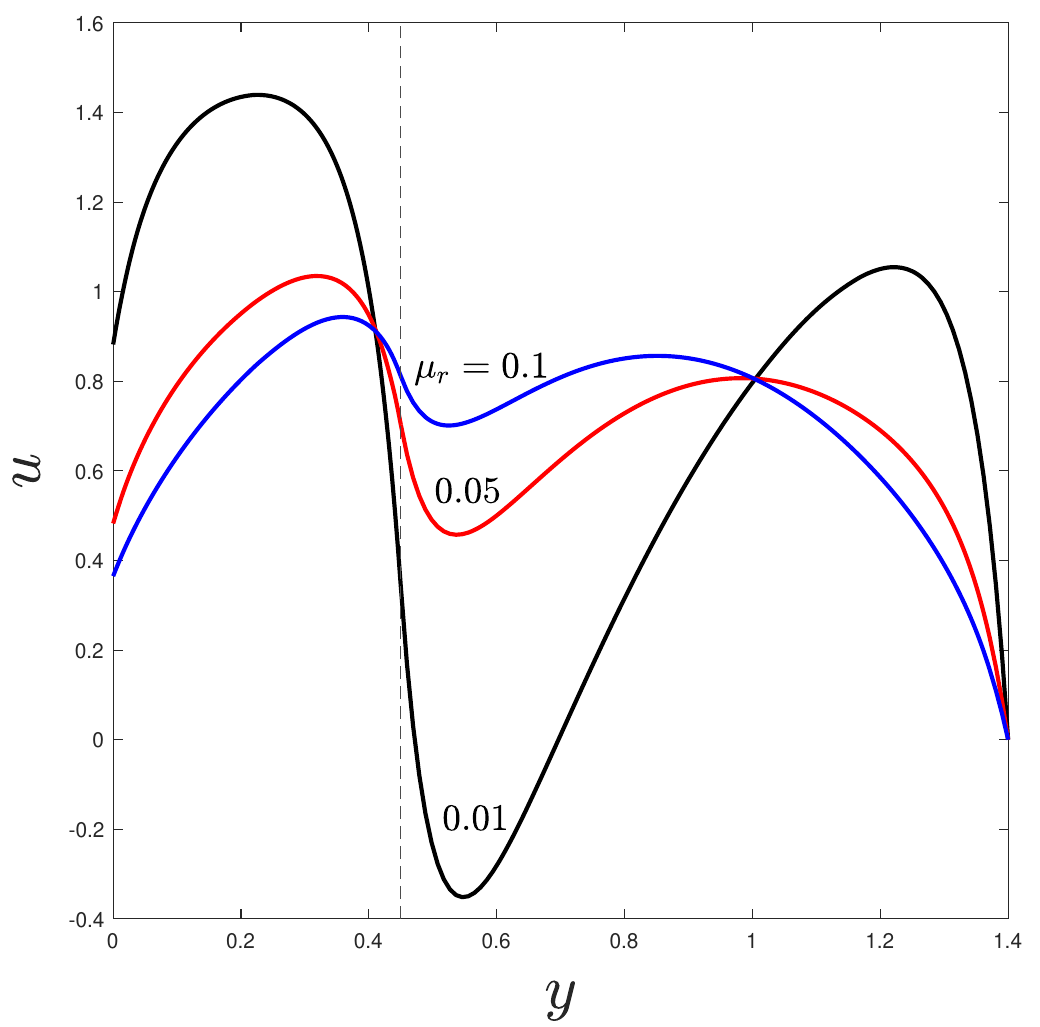}}
\caption{Axial composite velocity versus $y$ for (a) $a=0.34$, $0.4$,(b) $\lambda=1,1.5,1.75,2$ and (c) $\mu_r=0.01,0.05,0.1$ at $x=0$.}
\end{figure}
\subsection{Effect of pressure gradient and shear stress on pain realization}
\noindent
During and post-injection, a patient may realize swelling and pain from the injection site. According to some researchers, both the pressure gradient and shear stress act as an indicator of the evolvement of physical pain \citep{mueller2005pressure,goossens2009fundamentals}. Consequently, it is necessary to explore the nature of the pressure gradient $\left(\partial p / \partial x\right)$ and shear stress mainly at three different positions $y=y_{0}$ (line of injection), $y=0$ (SM interface) and $y=R(x)$ (SD interface). Figs. \ref{pressure_various_a}-\ref{pressure_various_mur} illustrate that the pressure gradient within $[-0.5,0.5]$ becomes symmetric about the line $x=0$. Consequently, one can pay attention to $\partial p / \partial x$ within $[0,0.5]$, which shows both monotonic and non-monotonic nature depending on the parameters. The non-monotonic nature of $\partial p / \partial x$ indicates the development of an adverse pressure gradient that causes eddy structure formation due to the flow separation. The adverse pressure gradient is developed mainly due to the exchange of kinetic energy between AC and IF due to significant IF viscosity variation caused after fluid injection. This adverse pressure gradient leads to the creation of secondary eddy within the lifted portion of the SCL. We have discussed earlier that these eddies are helpful for better blending of injected fluid-containing drugs within IF. The pressure gradient becomes maximum at the line of injection. The next possible location having a low magnitude of the pressure gradient is the SM interface. Therefore, the SD interface experiences the lowest gradient of pressure.\\
\begin{figure}[h!]
\centering
\subfigure[]{\label{pressure_various_a}\includegraphics[width=0.49\textwidth]{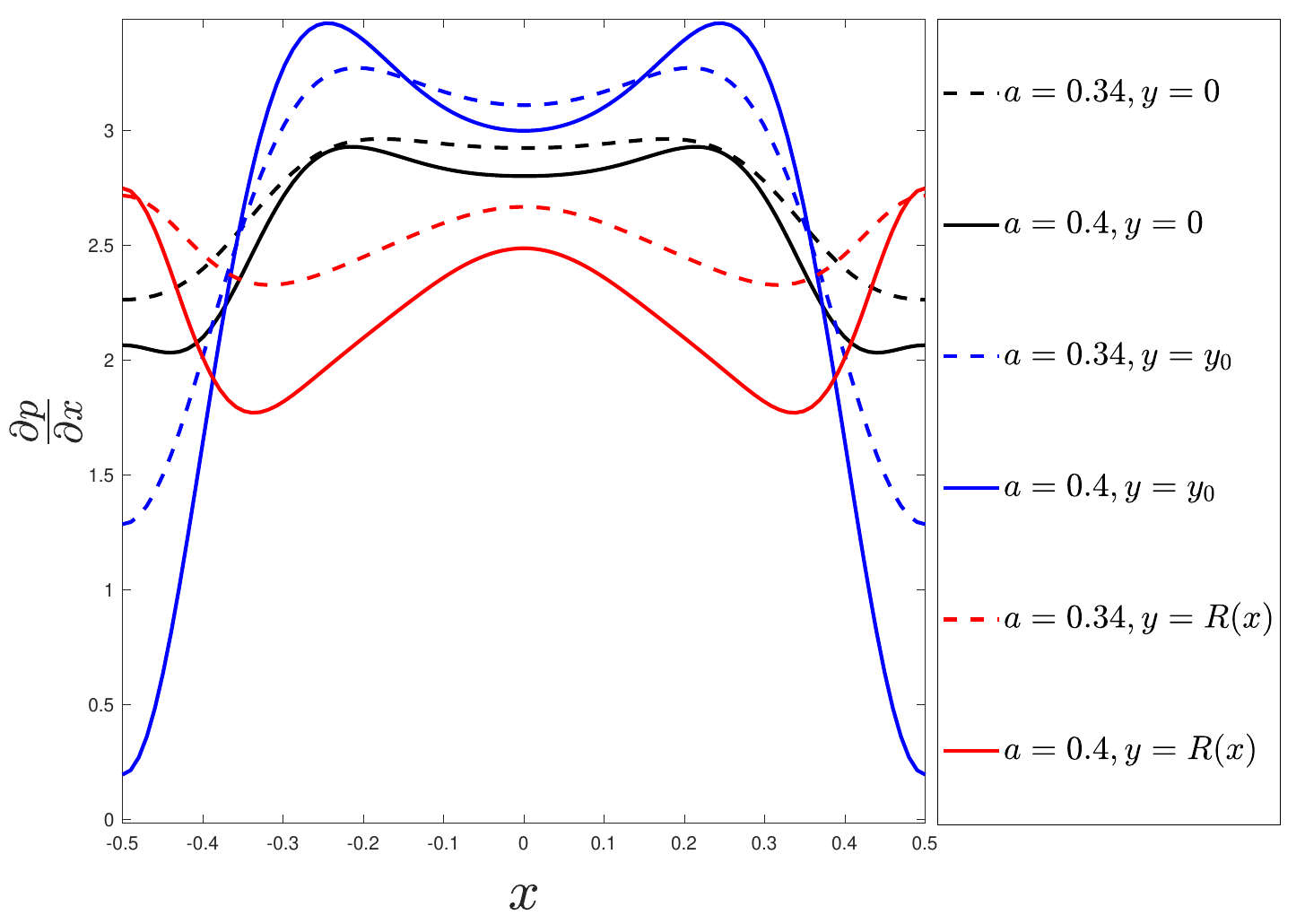}}
\subfigure[]{\label{pressure_various_lambda_best}\includegraphics[width=0.49\textwidth]{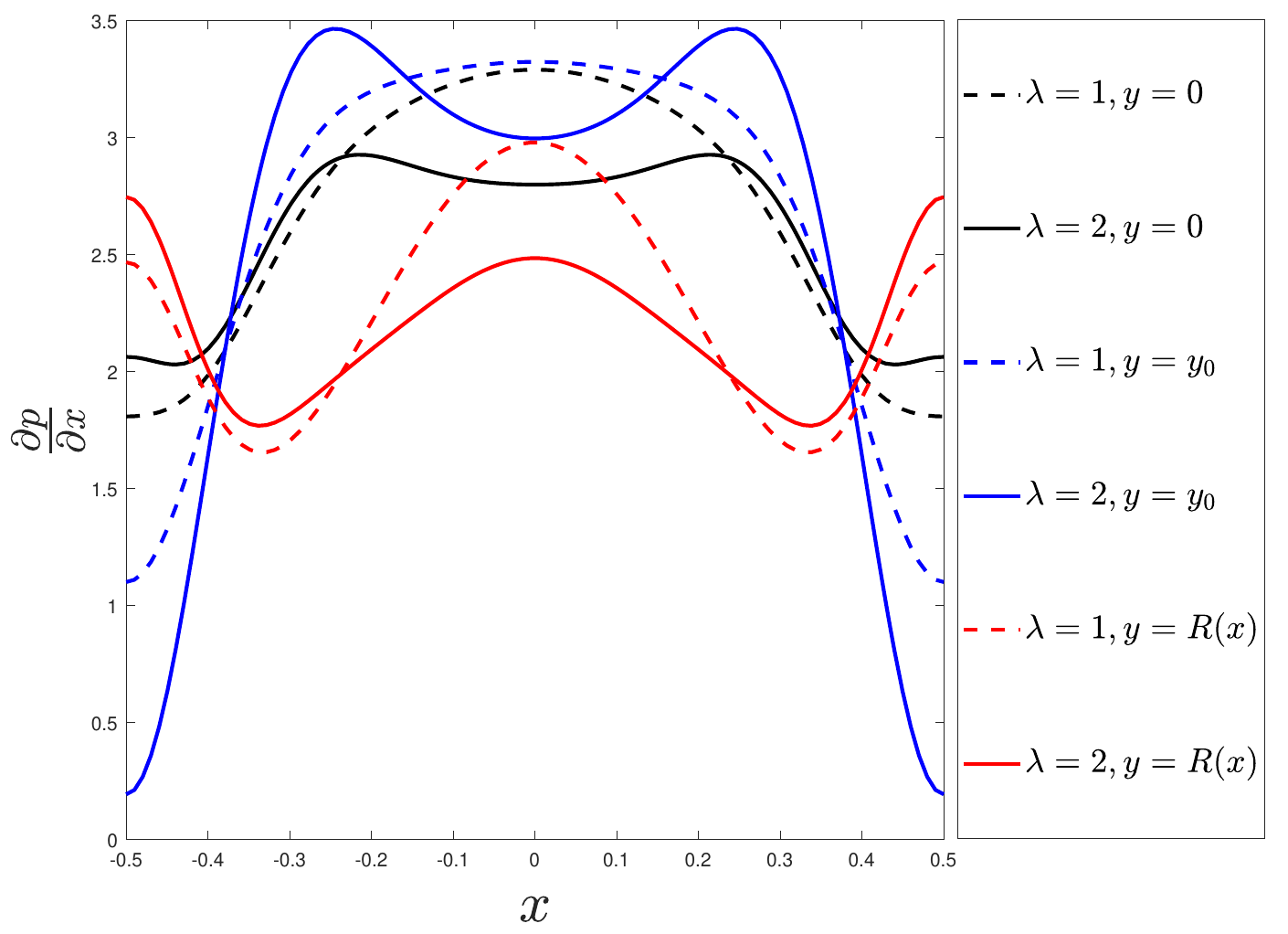}}
\subfigure[]{\label{pressure_various_mur}\includegraphics[width=0.49\textwidth]{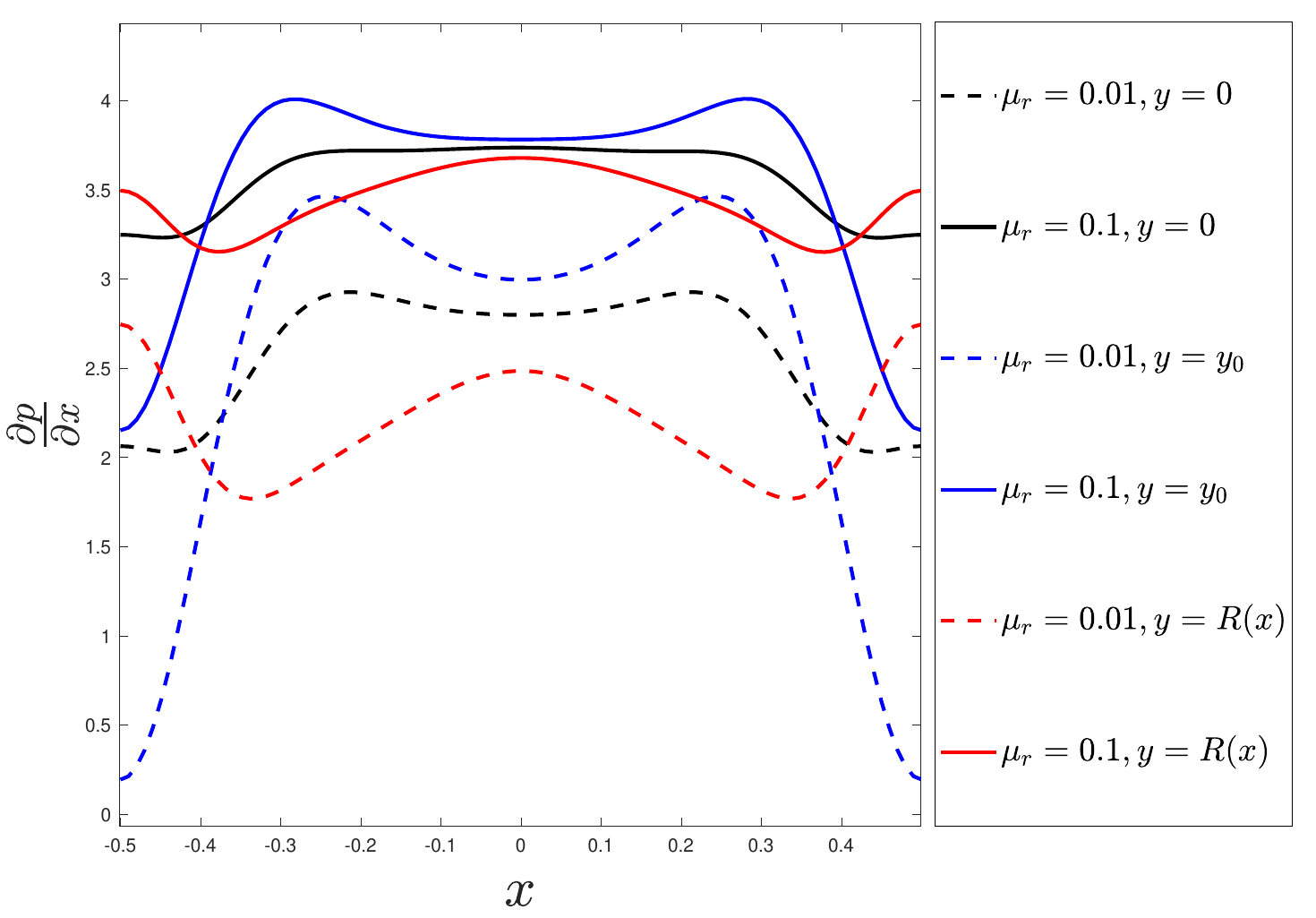}}
\caption{The pressure gradient variation with $x$, ($x\in[-0.5,0.5]$) corresponding to (a) $a=0.34, 0.4$ (b) $\lambda=1, 2$ (c) $\mu_r=0.01, 0.1$ at three locations $y=0$ (SM interface), $y=y_{0}$ (line of injection), and $y=R(x)$ (SD interface).}
\end{figure}

\begin{figure}[h!]
\centering
\subfigure[]{\label{shear_various_a}\includegraphics[width=0.48\textwidth]{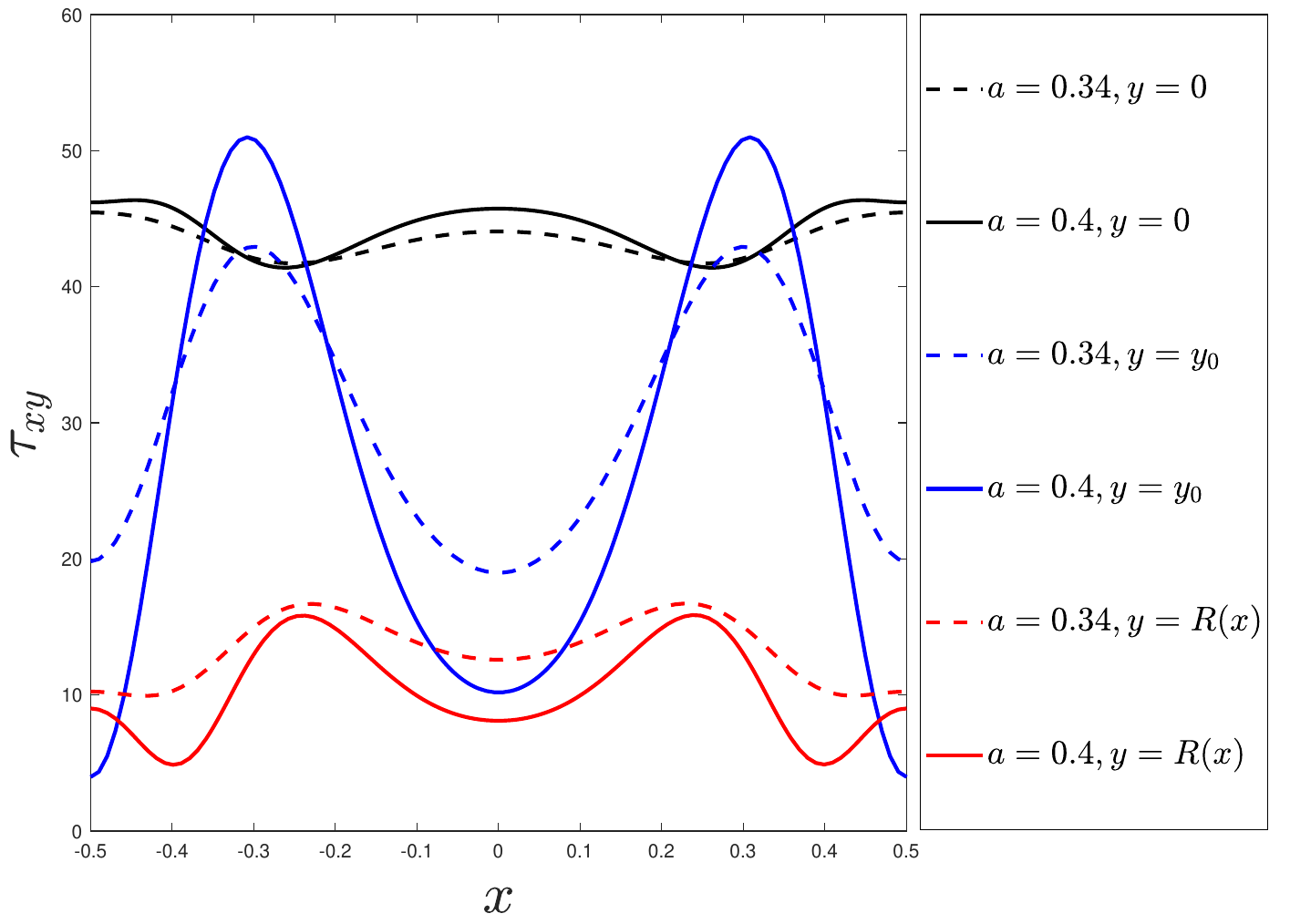}}
\subfigure[]{\label{shear_various_lambda_best}\includegraphics[width=0.48\textwidth]{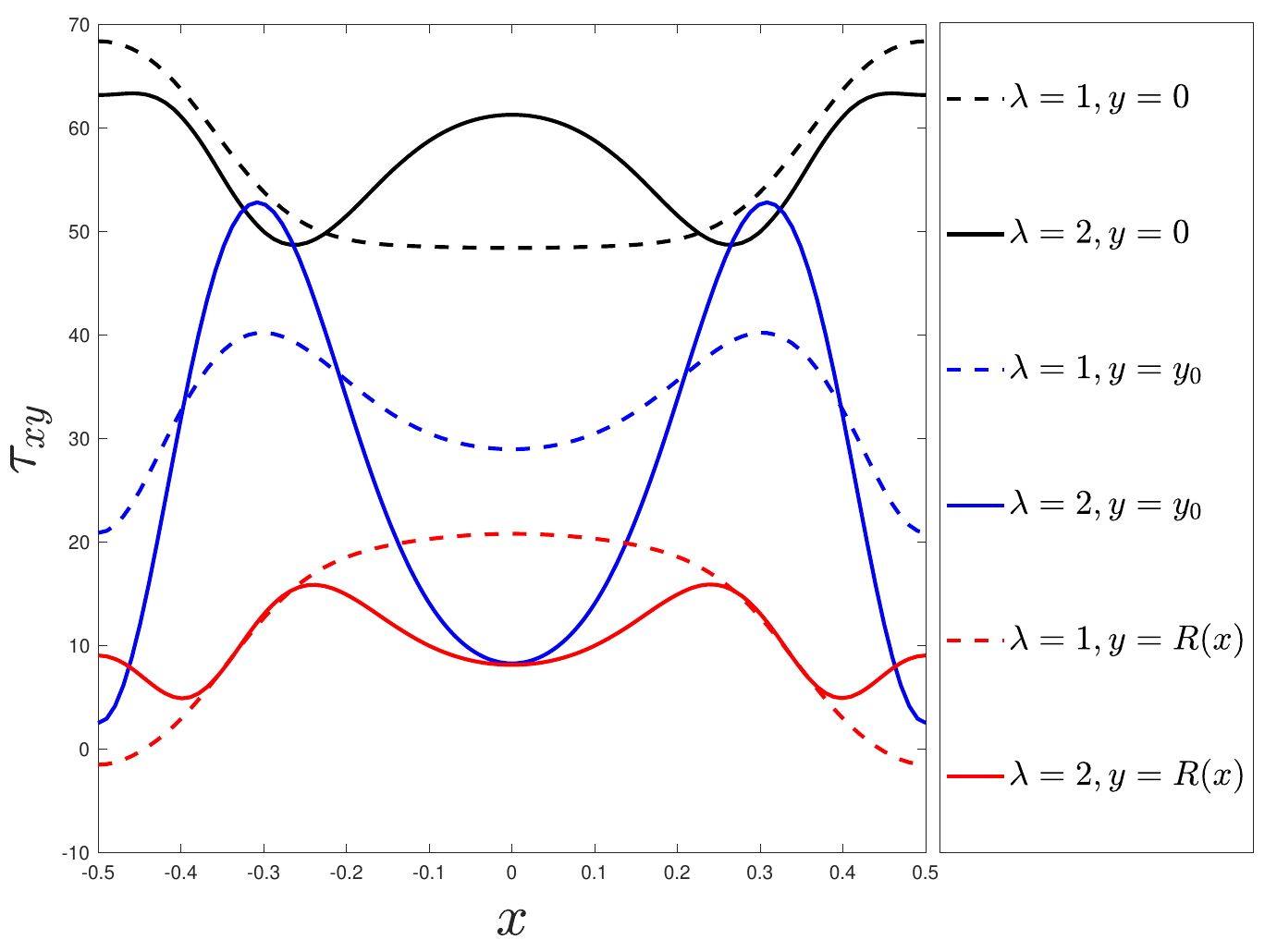}}
\subfigure[]{\label{shear_various_mur}\includegraphics[width=0.48\textwidth]{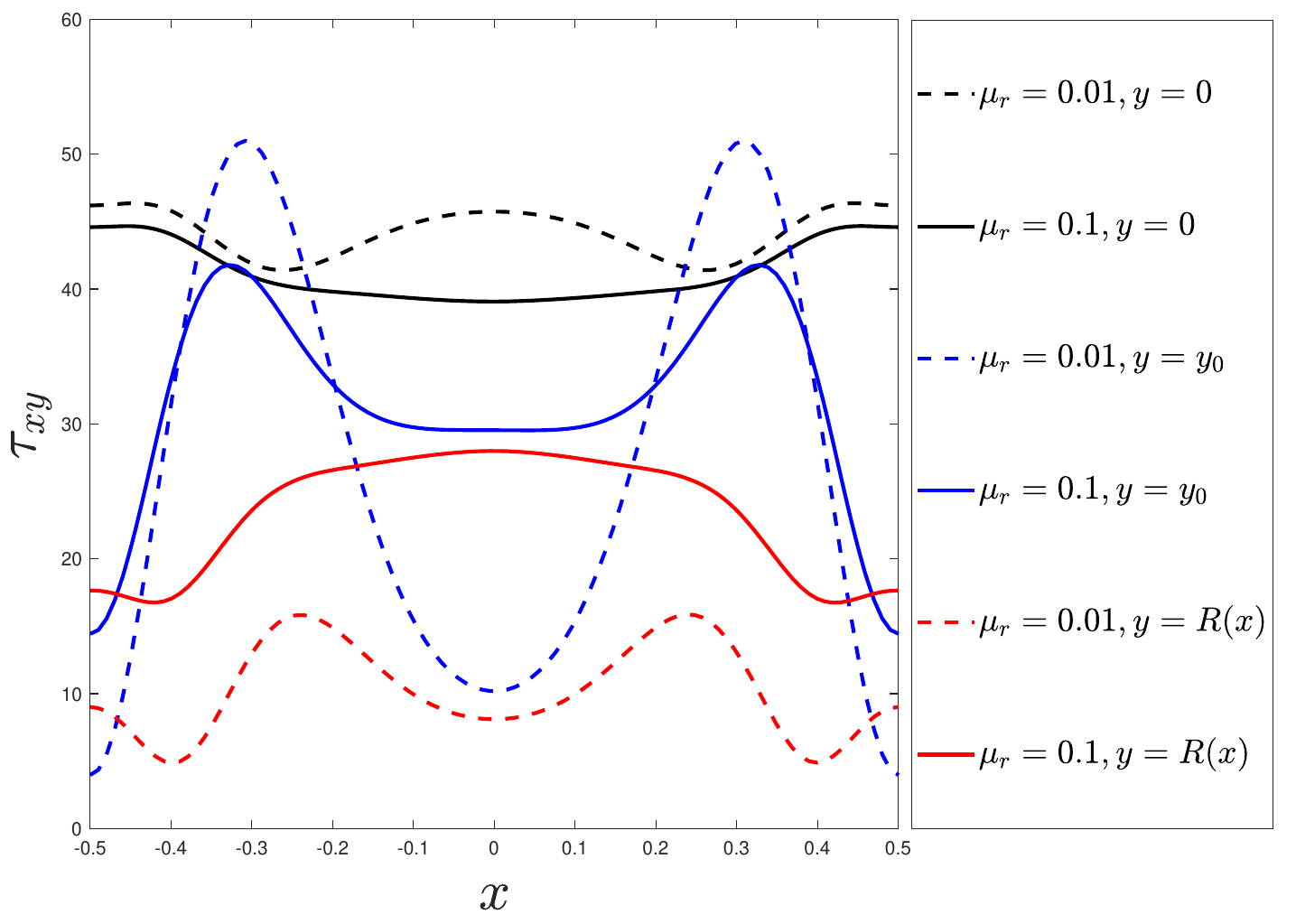}}
\caption{Shear stress profiles for (a) $a=0.34, 0.4$ (b) $\lambda=1, 2$ (c) $\mu_r=0.01, 0.1$ at three locations $y=0$ (SM interface), $y=y_{0}$ (line of injection), and $y=R(x)$ (SD interface).}
\end{figure}

\noindent
On the other hand, the calculated shear stress as given by
\begin{equation}\label{SEq11}
\tau_{xy} = \frac{\partial u}{\partial y} + \delta^2 \frac{\partial v}{\partial x},
\end{equation}
behaves non-monotonic for $x\in[-0.5, 0.5]$ at the three distinct positions $y=y_{0}$ (line of injection), $y=0$ (SM interface) and $y=R(x)$ (SD interface) (see Figs \ref{shear_various_a}-\ref{shear_various_mur}). Similar to the pressure gradient, as shown in the Figs \ref{shear_various_a}-\ref{shear_various_mur} the shear stress is also symmetric about $x=0$ within $[-0.5, 0.5]$. The non-monotonic behavior of shear stress tends to display large fluctuation for $x\in[-0.5,0.5]$ at $y=y_{0}$. But, such a high fluctuation never occurs for the other two positions of $y$. However, the SM interface experiences larger average shear stress than both the line of injection and SD interface. In addition, between SD and SM interfaces, the shear stress behaves exactly opposite each other. The behavioral difference of the shear stress fields between the SD interface and the line of injection allows the generation of secondary eddy from the primary one. On the other hand, the high shear stress close to the SM interface does not allow eddy formation. Hence, the formation of secondary eddy takes place near the SD interface only. Eventually, it can be said that the larger pressure gradient at $y=y_{0}$ and larger average shear stress at $y=0$ help the injected fluid to move away from the line of injection and consequently lateral spreading of injected fluid at the SM interface due to the consideration of slip property.\\
\begin{figure}[h!]
\centering
\includegraphics[width=0.5\linewidth]{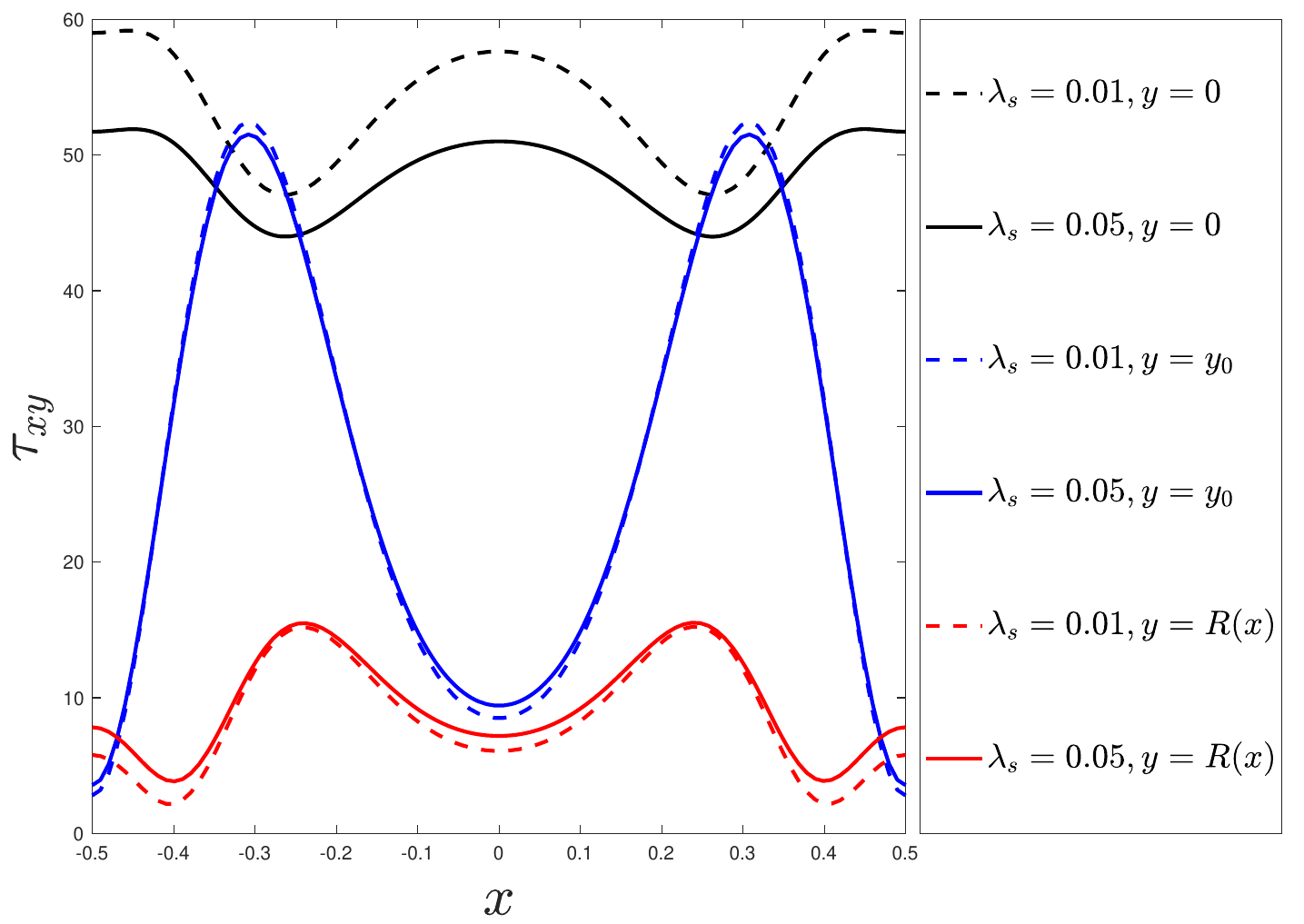}
\caption{Shear stress variation with $x\in[-0.5,0.5]$ for $\lambda_s=0.01, 0.05$ and at three different locations (i) SM interface ($y=0$) (ii) line of injection ($y=y_{0}$) (iii) SD interface ($y=R(x)$).}
\label{shear_various_slip}
\end{figure}

\noindent
The fluctuation in the non-monotonic profile of both the pressure gradient and shear stress profile corresponding to $y=y_{0}$ becomes pronounced with an increase in the depth of SCL pinched up. But at the other two positions, such high fluctuation cannot be seen. Moreover, it is noticed that pressure gradient and shear stress reduce with the increase in the depth of SCL pinched up during injection. However, it is rather difficult to predict anything from Figs. \ref{pressure_various_a} and \ref{shear_various_a} about the variations of $a$ on both the pressure gradient and shear stress at $y=0$ and $y=y_{0}$. On the other hand, the impact of $a$ on both the pressure gradient and shear stress can be easily understood at $y=R(x)$ (SD interface). Therefore at this moment, we may focus on the SD interface only. Now, suppose we associate the experience of the intensity of pain with the variation of pressure gradient and shear stress. In that case, we find that a patient may realize less pain from the SD interface for a higher depth of skin pinching during the injection.\\

\begin{figure}[h!]
\centering
\includegraphics[width=0.5\linewidth]{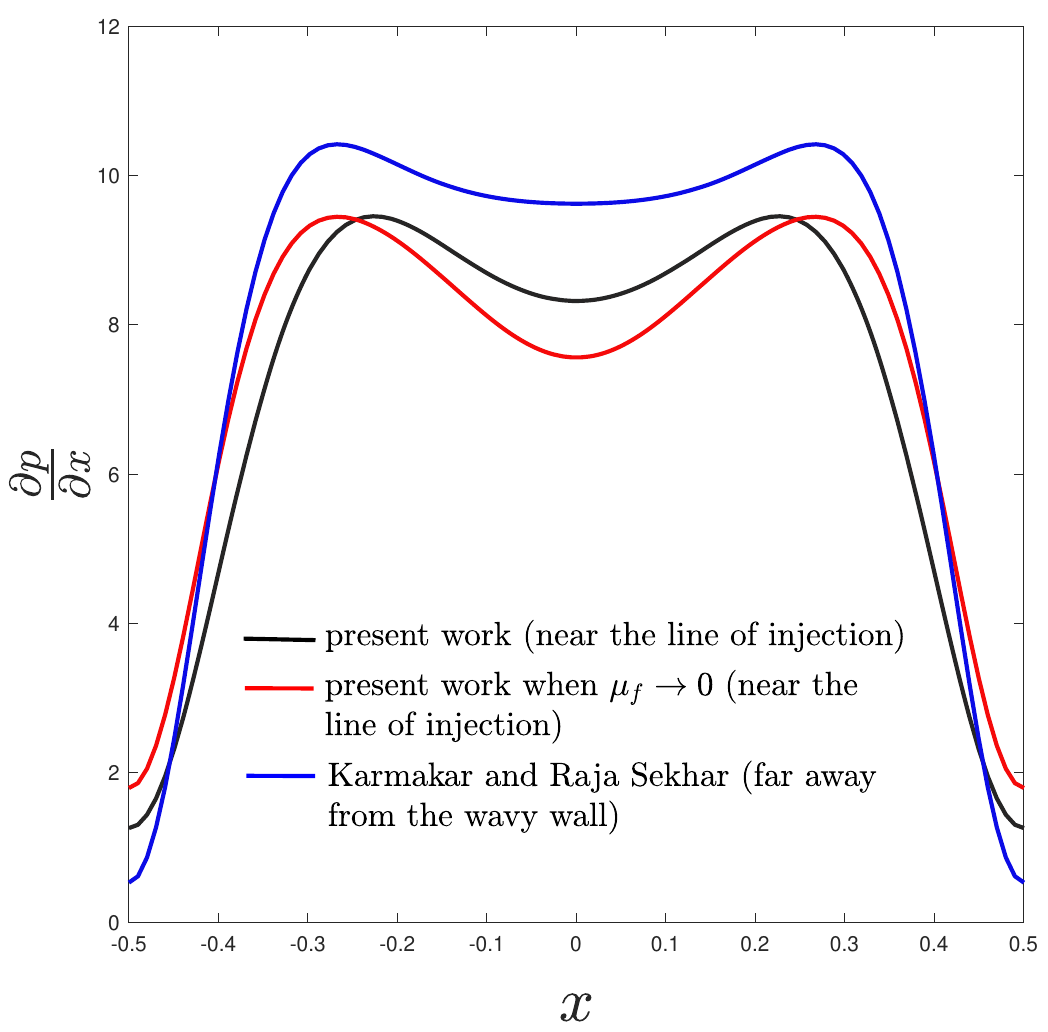}
\caption{Comparison of pressure gradients obtained (i) from the present study (black line) near the line $y=y_0$ (line of injection), (ii) in the limiting case $\mu_f\rightarrow0$ (red line) close to the line of injection (see Eq.(\ref{HEq8})) and (iii) from the study of \citet{karmakar2017note} which discusses the hydrodynamics of fluid flow through a wavy anisotropic porous channel (blue line).}
\label{compare}
\end{figure}
\noindent
The tissue anisotropy destroys the monotonic nature of the profiles of both the shear stress and pressure gradient as depicted through the Figs. \ref{pressure_various_lambda_best} and \ref{shear_various_lambda_best}. In other words, higher $\lambda$ (more precisely for all possible $\lambda>1$) imparts fluctuation in the pressure gradient and shear stress profiles. In particular for $\lambda=2$, the profiles of $\partial p/ \partial x$ at $y=y_{0}$ and $y=R(x)$ behave exactly opposite to each other. This opposite behavior of $\partial p/ \partial x$ is responsible for creation of secondary eddy from the primary one. In addition, we notice that at $y=0$ and $y=y_{0}$, $\partial p/ \partial x$ maintains monotonic nature for $\lambda=1$ (isotropic). This monotonic nature changes to non-monotonic with increased $\lambda$ beyond $1$ (adverse pressure gradient is developed). At the three positions, average magnitude of $\partial p/ \partial x$ reduces with increase in anisotropy. Hence, one would expect less pain generation from those sites due to increased anisotropy. On the other hand, the development of larger shear stress at the SM interface with increased $\lambda$ may consequence a patient to realize more pain generated from the SM interface though $\partial p/ \partial x$ reduces with increase in $\lambda$. Also at $y=y_{0}$, an increase in the shear stress is noted with the increased anisotropy. But at the SD interface, the average shear stress reduces with increase in $\lambda$ beyond $1$. Like $\partial p/ \partial x$, the shear stress changes its nature from monotonic to non-monotonic with increased tissue anisotropy. From the overall discussions, it is seen that both $\partial p/ \partial x$ and shear stress reduces with increase in $\lambda$ beyond the isotropic limit. Hence, a patient may experience less pain in the form of superficial pain from the SD interface when the SCL possesses significant larger anisotropy.\\

\begin{figure}[h!]
\centering
\subfigure[]{\label{pressure_optimum}\includegraphics[width=0.48\textwidth]{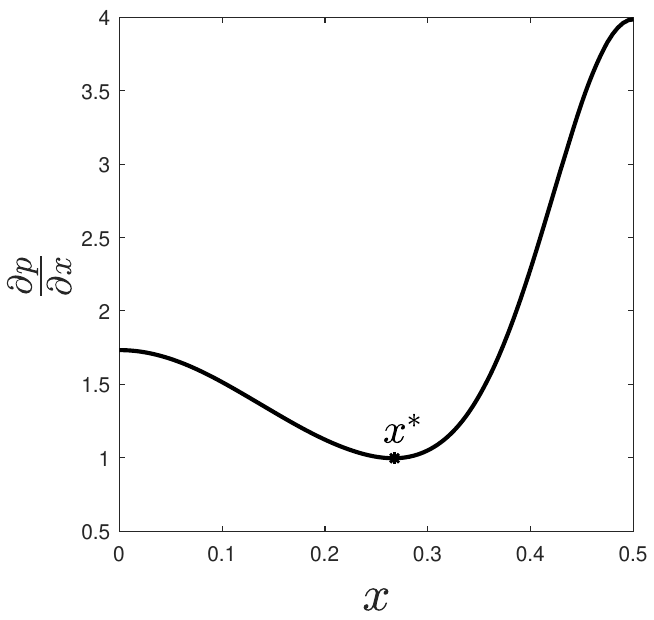}}
\subfigure[]{\label{lambda_vs_a}\includegraphics[width=0.48\textwidth]{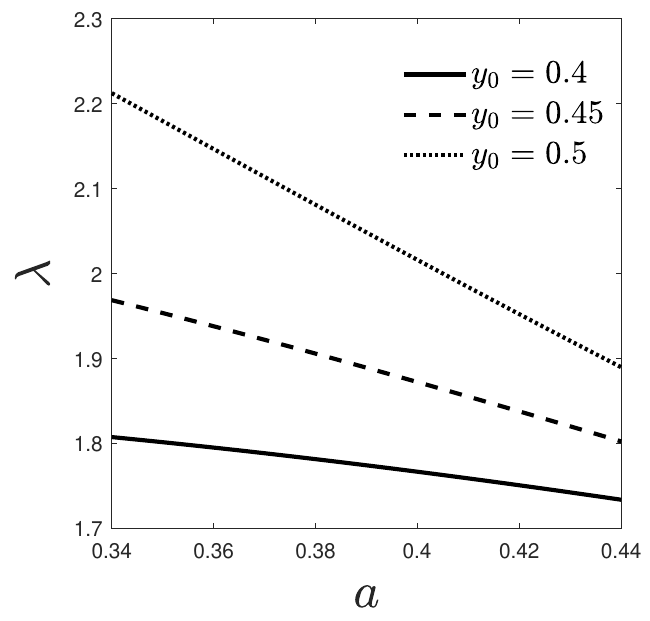}}
\caption{(a) Pressure gradient variation with respect to $x\in[0,0.5]$; $x^*$ is the optimum value point, (b) Anisotropic ratio ($\lambda$) versus $a$ for $y_0=0.4,0.45,0.5$.}
\end{figure}
\noindent
Both the pressure gradient and shear stress significantly vary with the viscosity of injected drug. Figs. \ref{pressure_various_mur} and \ref{shear_various_mur} illustrate pressure gradient and shear stress for $\mu_r=0.01$ and $\mu_r=0.1$ corresponding to three locations $y=0$, $y=y_0$ and $y=R(x)$ while other parameters are $a=0.4$, $\lambda=2$, $\lambda_s=0.05$ and $y_0=0.45$. We observe the rise of $\partial p/ \partial x$ with the increasing magnitude of $\mu_r$ in the above three locations of $y$. Since $\lambda=2$ (consideration anisotropic permeability of SCL), a non-monotonic nature of both $\partial p/ \partial x$ and shear stress is observed. This behavior becomes more aggressive with decreasing $\mu_r$. Therefore, upon injecting a low viscous fluid (much lower as compared to the AC), we can achieve such high non-monotonic behavior in $\partial p/ \partial x$. Moreover, corresponding to $\mu_r=0.01$, once again the opposite behavior of $\partial p/ \partial x$ is noticed between the positions $y=y_{0}$ and $y=R(x)$. This will guarantee the formation of secondary eddy from the primary one.\\

\begin{figure}[h!]
\centering
\subfigure[]{\label{normal_near_lambda}\includegraphics[width=0.48\textwidth]{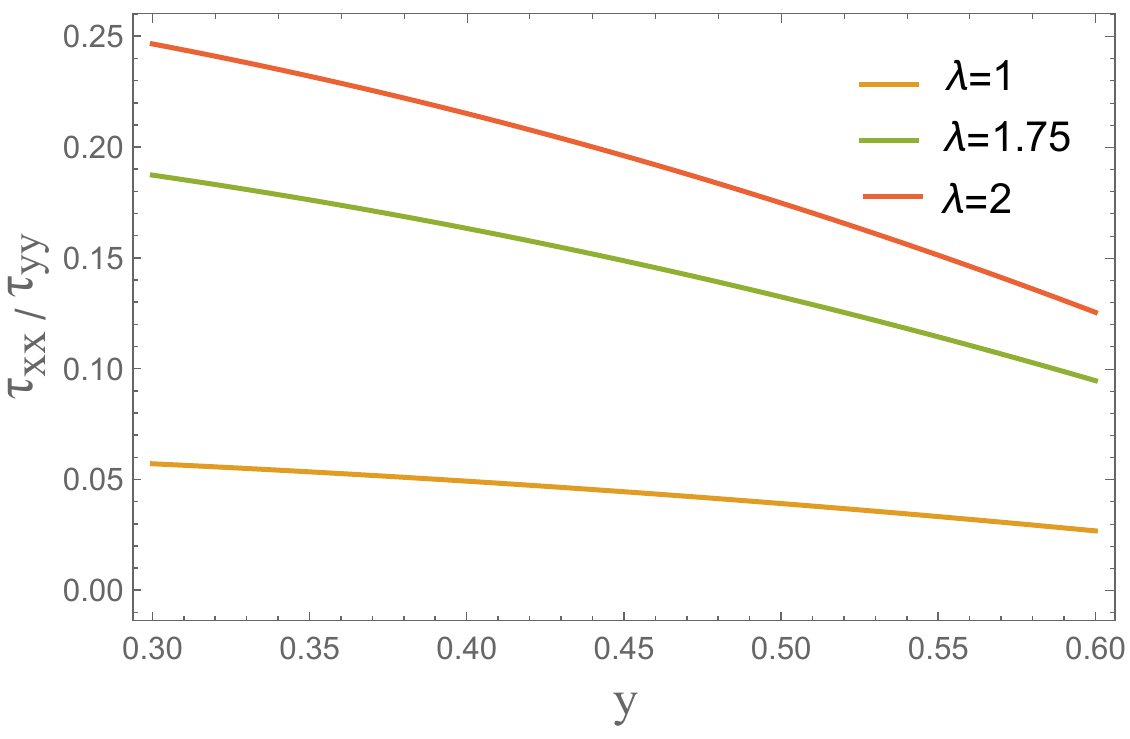}}
\subfigure[]{\label{shear_near_lambda}\includegraphics[width=0.46\textwidth]{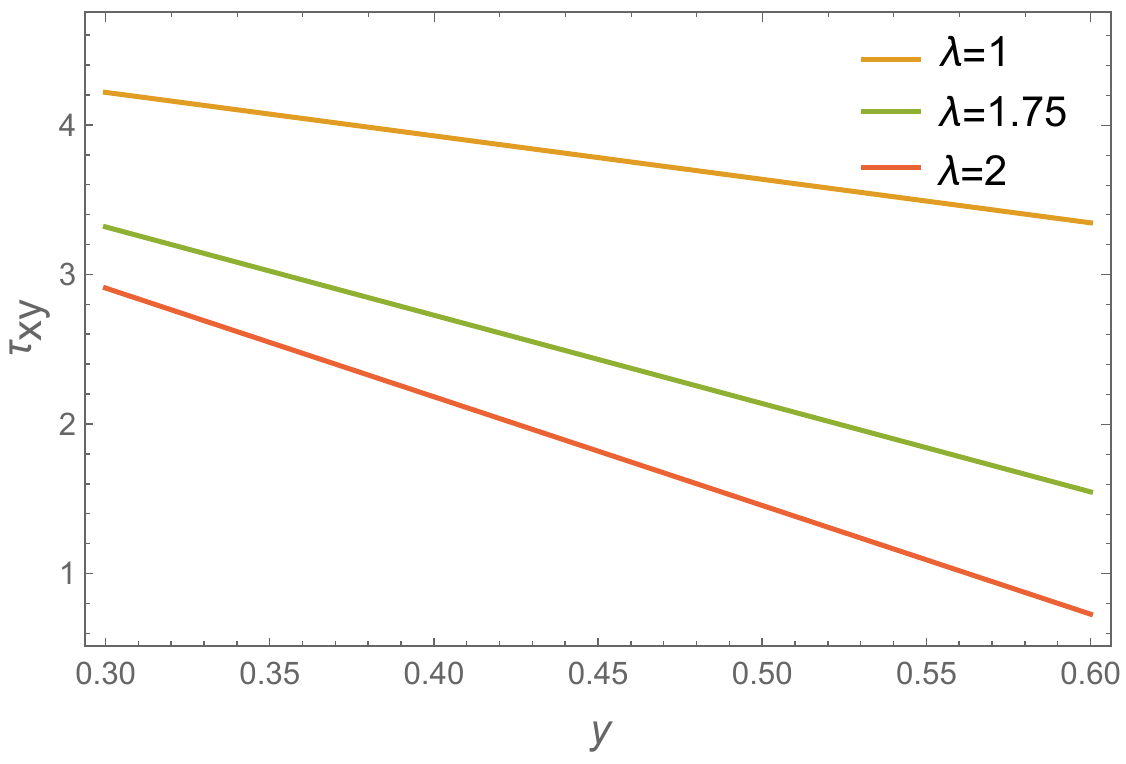}}
\subfigure[]{\label{shear_near_image}\includegraphics[width=0.49\textwidth]{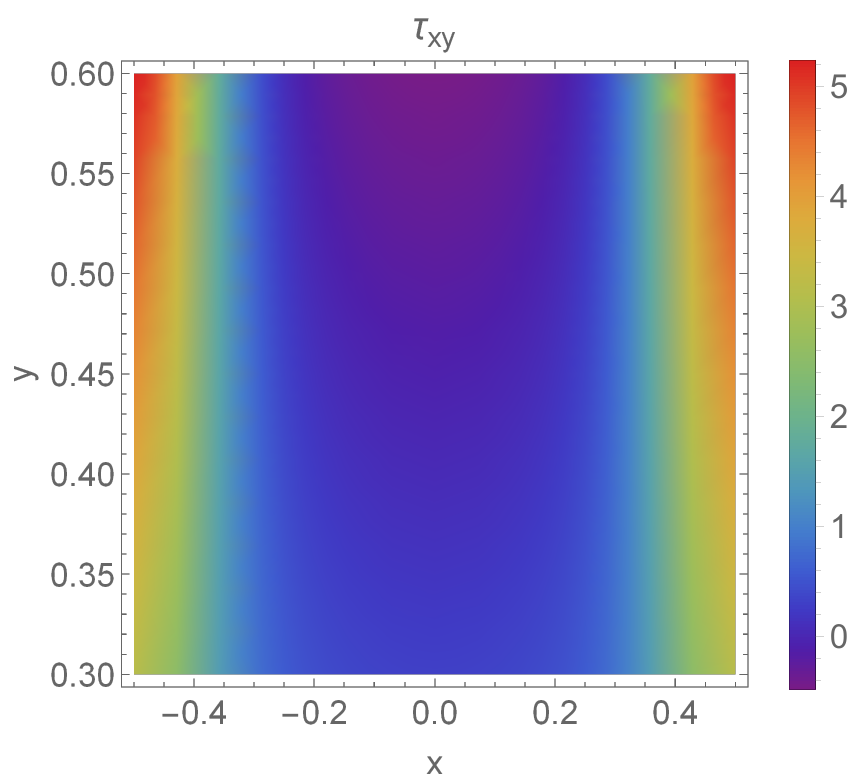}}
\caption{(a) Normal stress components $\tau_{xx}$ \& $\tau_{yy}$, (b) shear stress $\tau_{xy}$ for $\lambda=1,1.75,2$ and (c) shear stress field within the domain $\{(x,y): -0.5\leq x \leq 0.5, 0.3\leq y \leq 0.6 \}$ corresponding to $\lambda=2$.}
\end{figure}
\noindent
Further, we focus on the shear stress distribution for $\mu_r=0.01,0.1$ at the three locations $y=0$, $y=y_0$ and $y=R(x)$ when the other parameters are fixed as above. As analogous to the impact of pressure gradient, low $\mu_{r}$ shows sensitivity towards the shear stress field. The response of the shear stress field is not much significant near the SM interface (i.e, near $y=0$). However, near the line of injection (i.e, $y=y_{0}$) the increased non-monotonicity in the shear stress profile is noted due to reduced $\mu_{r}$. It is expected that this sensitivity (increase in non-monotonicity) becomes even higher in case of further lower $\mu_{r}$. On the other hand, an increased magnitude of $\partial p/ \partial x$ and shear stress can be noted at the SD interface with increase in $\mu_{r}$ from $0.01$ to $0.1$. This incident predicts the enhancement of superficial pain when a high viscous fluid is injected within SCL. Injection of low viscous fluid leads to both primary and secondary eddies as a consequence of the rapid change in the shear stress field between $y=0$ and $y=y_{0}$ and between $y=y_{0}$ and $y=R(x)$ respectively. Consequently, from the observations of pressure gradient and shear stress field we may conclude that besides the higher tissue anisotropy ($\lambda\geq1.75$), a larger difference in viscosity between the injected fluid and AC (i.e., $\mu_r<0.05$) causes adequate circulation of the injected fluid within the interstitial space of SCL.
\subsection{Effect of slip coefficient $\lambda_s$ on shear stress}
\noindent
Since the slip coefficient $\lambda_s$ is proportional to $\sqrt{Da}$, it characterizes the material property of the SM interface while responding to the generated shear stress during the fluid injection process. The effect of $\lambda_s$ on shear stress is distributed perpendicular to the other portions of SCL as well. With ageing SCL becomes thin and develops rigidity ($Da$ may assume smaller value). As a result, SM interface also attains higher rigidity. This suggests $\lambda_s$ to assume smaller value as compared to lesser rigid SM interface. On the other hand, a low/middle aged or healthy person develops relatively low shear stress at the SM interface. Consequently, $\lambda_s$ assumes its larger magnitude as compared to the case of a thin or aged person. Based on the choice of $Da$ in this present study, we need $\lambda_{s}>0$ and can be chosen $\lambda_{s}=0.01,0.05$. Fig. (\ref{shear_various_slip}) elucidates that the SM interface is more sensitive to the shear stress developed as compared to the other portion of SCL. As $\lambda_s$ controls the shear stress at the SM interface, high shear stress is noted at the SM interface at a low magnitude of $\lambda_s$ (e.g., $\lambda_{s}=0.01$). The developed shear stress gets dissipated along the interface when $\lambda_s$ is high (e.g., $\lambda_{s}=0.05$). In the other words, the SCL acts as a shear stress absorber when $\lambda_s$ is large and saves the underlying ML from possible damage due to the stress. Hence, at higher $\lambda_s$, a patient may realize less pain as a result of shear stress dissipation. Consequently, high $\lambda_s$ induces a favorable situation by reducing the shear stress during fluid injection. With the ageing, the subcutaneous injection losses some of its benefits and develop chance of lasting pain to the patient. However, to a healthy or a young aged people, the developed pain may be lasted for less time.
\subsection{Hydrodynamic behaviour near the line of injection}
\noindent
We anticipate that the viscous force becomes negligible near the line of injection within the SC layer due to the weakening impact of the hydrodynamic boundary layer. One can compare this situation with the significant domination of viscous forces close to the wall for a Poiseuille type flow within a channel or tube. Moreover, such domination is found to be relevant in the case of Poiseuille type flow in fluid overlying a porous medium \citep{hill2008poiseuille} and 2D flow through a wavy anisotropic porous channel \citep{karmakar2017note} where flows near the boundary obeys the Brinkman equation but are far from the boundary viscous forces become less effective and hence the flow satisfies Darcy equation. Consequently in this study, the second-order derivative terms in the momentum equation for the IF motion and the terms containing $\mu_f/\mu_c$ can be dropped. Subsequently, we obtain the following equations for the leading order:
\begin{equation}\label{HEq2}
(u_{f_0} - u_{c_0}) ~\thicksim - \frac{\partial p_0}{\partial x},
\end{equation}
and
\begin{equation}\label{HEq3}
(v_{f_0} - v_{c_0}) ~\thicksim   \frac{{\partial^2 p}_{0}}{{\partial x}^2} y.
\end{equation}
From the volumetric flow rate balance, one can obtain the pressure gradient for leading order as
\begin{equation}\label{HEq4}
p_{0x}~\thicksim \left(\frac{R(x)}{2}-\frac{1}{R(x)}\right).
\end{equation}
Similarly, we can obtain the \text{O}($\delta^2$) equations as
\begin{equation}\label{HEq5}
(u_{f_1}-u_{c_1})~\thicksim  \lambda^2 \frac{p_{0xxx}}{2}y^2 + d_1(x),
\end{equation}
and consequently, using the volumetric flow rate across the cross section, $d_1(x)$ is determined as
\begin{equation}\label{HEq6}
d_1(x) ~\thicksim - \lambda^2 \frac{p_{0xxx}}{6}R^{2}(x) - \frac{R(x)}{2}.
\end{equation}
Finally, we have the pressure gradient for \text{O}($\delta^2$) as
\begin{equation}\label{HEq7}
p_{1_x} ~\thicksim -\lambda^2 \frac{p_{0xxx}}{2}\left(y^2 - \frac{R(x)^2}{3}\right) + \frac{R(x)}{2}.
\end{equation}
Thus, the pressure gradient up to \text{O}($\delta^2$) is
\begin{equation}\label{HEq8}
\frac{\partial p}{\partial x} ~\thicksim \frac{\partial p_0}{\partial x} - \delta^2 \lambda^2 \frac{p_{0xxx}}{2}\left(y^2 - \frac{R(x)^2}{3}\right) + \frac{\delta^2}{2}R(x).
\end{equation}
The utility of the above analysis aims to deduce the normal and tangential stresses in a convenient way (without involving tedious calculations). Hence to check the validity, we can compare the pressure gradient as in Eq. (\ref{HEq8}) with that of obtained from the main calculations near the line of injection (or far away from the interfaces). Therefore through Fig. \ref{compare}, we observe that both the pressure gradients show similar behavior both qualitative and quantitative points of view (black and red lines). We can also establish the validity of the present model by plotting the pressure gradient obtained from the study of \citet{karmakar2017note} corresponding to the fixed anisotropy ratio $\lambda=2$. The blue line represents the pressure gradient variation versus $x$ showing a nice qualitative agreement with the present study. Note that the anisotropic geometry is the key feature of both the studies.\\

\noindent
From the Eq. (\ref{HEq8}), we have the pressure gradient near the line injection. In particular, at the line of injection
\begin{equation}\label{oth3}
\frac{\partial p}{\partial x} = p_{0x} - \delta^2 \lambda^2 \frac{p_{0xxx}}{2}\left(y_0^2 - \frac{R(x)^2}{3}\right) + \frac{\delta^2}{2}{R(x)}.
\end{equation}
Pressure gradient from Eq. (\ref{oth3}) may attain its optimum value at some points which can be calculated. If $x^*$ denote such an optimum point (see Fig. (\ref{pressure_optimum})), then we have
\begin{equation}\label{oth4}
\lambda^2 \delta^2 \left(\frac{p_{0xxxx}(x^*)}{2}\left(y_0^2 - \frac{R(x^*)^2}{3}\right) - \frac{p_{0xxx}(x^*)}{3} R\left(x^{*}\right) R'(x^{*}) \right) = p_{0xx}(x^*) + \frac{\delta^2}{2}{R'(x^*)}
\end{equation}
Eq. (\ref{oth4}) gives an explicit relation between the anisotropic ratio $\lambda$, amplitude $a$ and penetration depth $y_0$ (height of the line of injection from SM interface). As tissue anisotropy is a material property, it can not be controlled from outside by an administrator during injection, though the pinching depth $a$ and penetration depth $(1-y_0)$ can be regulated during the injection time. So it would be very effective if we find a relation between $\lambda$, $a$ and $y_0$. In this context, Eq. (\ref{oth4}) gives an idea about $\lambda$ in terms of $a$ and $y_0$. Consequently, Fig. \ref{lambda_vs_a} shows variation of anisotropic ratio $\lambda$ with respect to amplitude $a$ corresponding to the different penetration depth $y_0=0.4, 0.45$ and $0.5$. It is evident that with the increase of $a$, $\lambda$ decreases along the injection line $y_0$. Even when skin pinching is high, and tissue anisotropy is relatively low, eddies can still form and assist in adequately mixing the injected fluid with IF. Moreover, tissue anisotropy increases as the injection line moves deeper into the tissue. This implies that higher injection depths may necessitate higher levels of tissue anisotropy compared to other cases.\\

\noindent
Fig. \ref{normal_near_lambda} represents normal stresses $\tau_{xx}$ and $\tau_{yy}$ for various $\lambda$ within the range $0.3\leq y \leq 0.6$ enclosing the line of injection. We observe that the normal stresses increase with $\lambda$. Therefore both along the $x$ and $y$ directions, the intensity of the pain at the injection site increases with tissue anisotropy. On the other hand, Fig. \ref{shear_near_lambda} shows opposite behaviour corresponding to that of shear stress variation concerning $\lambda$ within the above range of $y$. The pain generated due to the shear stress decreases with increasing tissue anisotropy. The inverse behavior of normal and shear stress fields maintains the mechanical equilibrium within the SC layer which becomes imbalanced due to the tissue anisotropy variation. The behavior of the shear stress field can be justified from the longitudinal pressure gradient ($\partial p/\partial x$) variations concerning $y$ for various $\lambda$ near the line of injection. The rate of shear stress variation along the direction normal to the line of injection is proportional to $\partial p/\partial x$. Moreover from Fig. \ref{shear_near_image}, one can notice the behaviour of shear stress in the domain $\{(x,y): -0.5\leq x \leq 0.5, 0.3\leq y \leq 0.6 \}$ containing the line of injection. It shows that the magnitude of shear stress is symmetric with respect to $x$-axis. The minimum value of shear stress attained at $x=0$. In other words, shear stress attains its minimum vale at the position where the skin is being lifted up. The injected drug spreads from the injection point to the extracellular region. Consequently, shear stress is developed away from the line $x=0$.

\section{Concluding Remarks}
\noindent
A two-dimensional fluid injection model within the subcutaneous layer (SCL) has been investigated using biphasic mixture theory. We study the flow pattern of the injected fluid (containing drug), pressure gradient, and shear stress in terms of parameters $a$ (skin pinching height), $\lambda$ (anisotropy ratio), $\mu_r$ (viscosity ratio), and $\lambda_S$ (slip coefficient). The overall hydrodynamic analysis reveals the creation of the primary eddy structure at the line of injection due to the development of a high-pressure gradient. In addition, it is seen that the lifted portion of the SCL may witness secondary eddies depending on the viscosity of the injected fluid and the height of the skin pinching. Depending on the anisotropy ratio, one may see this secondary eddy if the SD interface is not regular. Both the primary and secondary eddies play a significant role to homogenize the injected fluid with the IF when (i) the anisotropy ratio of SCL is greater than 1.5 (ii) low viscosity ratio ($\mu_{r}<0.05$) and (iii) high skin pinching height ($a>0.3$).  \\

\noindent
Since pressure gradient and shear stress act as possible indicators of pain generation inside SCL \citep{mueller2005pressure,goossens2009fundamentals}, the high viscous injected fluid may induce pain to a patient receiving SCI. On the other hand, our analysis reveals that low skin pinching height and small anisotropy ratio of SCL can be held responsible for realizing of more pain. Moreover, the shear stress at the SM interface is high corresponding to a small slip coefficient ($\lambda_{s}=0.001$). Thus, the SCL can affect the underlying ML by imparting high shear stress generated from the fluid injection. With the ageing issue, enhanced shear stress is developed at the SM interface which is manifested by a smaller slip coefficient ($\lambda_{s}\leq0.001$).\\
\section*{Acknowledgements}
\noindent
First author A. S. Pramanik acknowledges University Grants Commission (UGC), Govt. of India for providing junior research fellowship (NET-JRF, Award Letter Number: 1149/(CSIR-UGC NET DEC-2018)). Second author B. Dey acknowledges University research assistance (Ref. 1516/R-2020 dated 01.06.2020) of University of North Bengal for supporting this work. Authors are very much thankful to Prof. G. P. Raja Sekhar, Professor of Mathematics at Indian Institute of Technology Kharagpur, West Bengal, India, for his inspiration and necessary comments to improve this manuscript.
\section*{Appendix A. Solution to the leading-order problem}
\label{appendix1}
\noindent
With the boundary conditions, the solution of the equations (\ref{Eq1})-(\ref{Eq6}) are
$$u_{f_0}(x,y)=\frac{X_0(x,y) + Y_0(x,y)}{1+\mu_r}\hspace{0.2cm}, \hspace{0.2cm} u_{c_0}(x,y)=\frac{X_0(x,y) - \mu_rY_0(x,y)}{1+\mu_r} ,$$
where $X_0(x,y)$ and $Y_0(x,y)$ are given by\\
\[X_{0}(x,y) = \left\{
  \begin{array}{lr}
L_1 {p_{0}}_{x}y^2 + A_1^{(0)}(x) y + A_{2}^{(0)}(x) ,~~~~~~~~~~~~~~~~~~~~~~~~~ \text{if  } 0<y<y_{0}\\\\
L_1 {p_{0}}_{x}\left(R(x)-y\right)^2 + A_{3}^{(0)}(x) (R(x)-y) + A_{4}^{(0)}(x) ,  \text{if  } y_0<y<R(x)
\end{array}
\right.
\]\\
\[Y_{0}(x,y) = \left\{
  \begin{array}{lr}
B_1^{(0)}(x) \text{cosh}\left(\beta y\right) + B_2^{(0)}(x) \text{sinh}\left(\beta y\right) + L_2 p_{0x} ,~~~~~~~ \textrm{if} 0<y<y_{0}\\\\
B_3^{(0)}(x) \text{cosh}\left(\beta (R(x)-y)\right) + B_4^{(0)}(x) \text{sinh}\left(\beta (R(x)-y)\right) + L_2 p_{0x},\\
~~~~~~~~~~~~~~~~~~~~~~~~~~~~~~~~~~~~~~~~~~~~~~~~~~~~~~~~~~~~~~~~~ \text{if  } y_0<y<R(x)
\end{array}
\right.
\]
 in which $A_i^{(0)}(x)$, $B_i^{(0)}(x)$ ($i=1,2,3,4$) are constants of integration which can be calculated using the boundary conditions with the condition of continuity at $y=y_0$. $L_1$, $L_2$ and $\beta$ are given by the followings
$$L_1=\frac{\mu_r \alpha^2}{2} \text{      ,       } L_2 = -\frac{(\phi_f - \phi_c \mu_r)}{1+\mu_r} \text{      and       } \beta^2 = (1+\mu_r)\alpha^2.$$
Also using the equation of continuity, we have\\
$$v_{f_0}(x,y)=\frac{V_0(x,y) + W_0(x,y)}{1+\mu_r}\hspace{0.2cm}, \hspace{0.2cm} v_{c_0}(x,y)=\frac{V_0(x,y) - \mu_rW_0(x,y)}{1+\mu_r} ,$$
where $V_0(x,y)$ and $W_0(x,y)$ are given by\\
\[V_0(x,y) = \left\{
  \begin{array}{lr}
 - \frac{1}{3}L_1 p_{0xx}y^3 - \frac{1}{2}\left(A^{(0)}_{1}(x)\right)_{x} y^2 - \left(A^{(0)}_{2}(x)\right)_{x} y + c_1^{(0)}(x),\\
 ~~~~~~~~~~~~~~~~~~~~~~~~~~~~~~~~~~~~~~~~~~~~~~~~~~~~~~~~~~~~~~~~~~~~\text{if  } 0<y<y_{0}\\\\
 - \frac{1}{3}L_1 p_{0xx}\left(R(x)-y\right)^3 - \frac{1}{2}\left(A^{(0)}_{3}(x)\right)_{x} (R(x)-y)^2 \\
 - \left(A^{(0)}_{4}(x)\right)_{x} (R(x)-y) + c^{(0)}_2(x), ~~~~~~~~~~~~~~~~~~~~~~~\text{if  } y_0<y<R(x)
\end{array}
\right.
\]\\
\[W_0(x,y) = \left\{
  \begin{array}{lr}
  - \frac{1}{\beta}\left(B^{(0)}_1(x)\right)_{x} \text{sinh}\left(\beta y\right) - \frac{1}{\beta}\left(B^{(0)}_2(x)\right)_{x} \text{cosh}\left(\beta y\right) - L_2 p_{0xx} y + d_1^{(0)}(x),\\
  ~~~~~~~~~~~~~~~~~~~~~~~~~~~~~~~~~~~~~~~~~~~~~~~~~~~~~~~~~~~~~~~~~~~~~\text{if  } 0<y<y_{0}\\\\
  - \frac{1}{\beta}\left(B^{(0)}_3(x)\right)_{x} \text{sinh}\left(\beta (R(x)-y)\right) - \frac{1}{\beta}\left(B^{(0)}_4(x)\right)_{x} \text{cosh}\left(\beta (R(x)-y)\right)\\ - L_2 p_{0xx} (R(x)-y)  + d_2^{(0)}(x), ~~~~~~~~~~~~~~~~~~~~~~~~~~~~~~~\text{if  } y_0<y<R(x)
 \end{array}
\right.
\]
in which $c_i^{(0)}(x)$, $d_i^{(0)}(x)$ ($i=1,2$) are constants of integration which can be calculated using the boundary conditions.
\section*{Appendix B. Solution to the \text{O}($\delta^2$) problem}
\label{appendix2}
The general solution of the \text{O}($\delta^2$) problem is
$$u_{f_1}(x,y)=\frac{X_1(x,y) + Y_1(x,y)}{1+\mu_r}\hspace{0.2cm}, \hspace{0.2cm} u_{c_1}(x,y)=\frac{X_1(x,y) - \mu_rY_1(x,y)}{1+\mu_r} ,$$
where $X_1(x,y)$ and $Y_1(x,y)$ are given by
\[X_1(x,y) = \left\{
  \begin{array}{lr}
  E_1(x)y^4 + E_2(x)y^3 + E_3(x) y^2 + A^{(1)}_1(x) y^2 + A^{(1)}_2(x)y + A^{(1)}_3(x),\\
  ~~~~~~~~~~~~~~~~~~~~~~~~~~~~~~~~~~~~~~~~~~~~~~~~~~~~~~~~~~~~~~~~~~~~\textrm{if} 0<y<y_{0} \\\\
  E_1(x)(R(x)-y)^4 + E_4(x)(R(x)-y)^3 + E_5(x) (R(x)-y)^2 \\ + A^{(1)}_1(x) (R(x)-y)^2  + A^{(1)}_{4}(x)(R(x)-y) + A^{(1)}_5(x), ~ \text{if~~~~} y_0<y<R(x)
 \end{array}
\right.
\]\\
\[Y_1(x,y) = \left\{
  \begin{array}{lr}
  B^{(1)}_1(x) + B^{(1)}_2(x)\textrm{cosh}(\beta y) + B^{(1)}_3(x)\textrm{sinh}(\beta y) + F_1(x) y \textrm{sinh}(\beta y) \\ +
  F_2(x) y \textrm{cosh}(\beta y) + F_3(x) (\beta^2 y^2 + 2) + F_4(x)y + F_{5}(x),~~\text{if  } 0<y<y_{0}\\\\
  B^{(1)}_1(x) + B^{(1)}_4(x)\textrm{cosh}(\beta (R(x)-y)) + B^{(1)}_5(x)\textrm{sinh}(\beta (R(x)-y)) \\
  + F_{6}(x) (R(x)-y) \textrm{sinh}(\beta (R(x)-y)) + F_{7}(x) (R(x)-y) \textrm{cosh}(\beta (R(x)-y)) \\
  + F_3(x) (\beta^2 (R(x)-y)^2 + 2) + F_{8}(x)(R(x)-y) + F_{5}(x), ~\text{if  } y_0<y<R(x)
 \end{array}
\right.
\]
in which $A_i^{(1)}(x)$, $B_i^{(1)}(x)$ ($i=1,2,3,4,5$) are constants of integration which can be calculated using the boundary conditions and condition of continuity at $y=y_0$. $E_{i}(x)$($i=1,2,3,4,5$) and $F_{i}(x)$ ($i=1,2,3,4,5,6,7,8$) are explicitly given in the Appendix C. \\

\noindent
Hence, the general solution of the considered problem is determined upto \text{O}($\delta^2$) as
$$u_{f}(x,y)=u_{f_0}(x,y) + \delta^2 u_{f_1}(x,y) + \text{O}(\delta^4), $$
$$u_{c}(x,y)=u_{c_0}(x,y) + \delta^2 u_{c_1}(x,y) + \text{O}(\delta^4). $$
\section*{Appendix C.}
\label{appendix3}
\noindent
$E_1(x)=-\frac{1}{6}L_{1} p_{0xxx}$ , \hspace{0.5cm}  $E_2(x)=-\frac{1}{3} (A_1^{(0)}(x))_{xx}$, \hspace{0.5cm} $E_3(x)=- (A_2^{(0)}(x))_{xx}$ , \\\\
$E_4(x)=-\frac{1}{3} (A_3^{(0)}(x))_{xx}$, \hspace{0.5cm} $E_5(x)=- (A_4^{(0)}(x))_{xx}$ , \hspace{0.5cm} $F_1(x)=\frac{(\lambda^2-2)}{2\beta}(B_1^{(0)}(x))_{xx}$, \\\\
$F_2(x)=\frac{(\lambda^2-2)}{2\beta}(B_2^{(0)}(x))_{xx}$ , \hspace{0.5cm} $F_3(x)=\frac{\lambda^2}{2\beta^2} L_{2}p_{0xxx}$ , \hspace{0.5cm} $F_4(x)=\lambda^2 (d_1^{(0)}(x))_{x}$ ,\\\\
$F_{5}(x)= -\frac{2}{\beta^2}L_{2} p_{0xxx}$ , \hspace{0.5cm} $F_{6}(x)=\frac{(\lambda^2-2)}{2\beta}(B_3^{(0)}(x))_{xx}$ , \hspace{0.5cm} $F_{7}(x)=\frac{(\lambda^2-2)}{2\beta}(B_4^{(0)}(x))_{xx}$, \\\\
$F_{8}(x)=\lambda^2 \left(d_2^{(0)}(x)\right)_{x}$ .

\begin{thebibliography}{69}
\providecommand{\natexlab}[1]{#1}
\providecommand{\url}[1]{\texttt{#1}}
\expandafter\ifx\csname urlstyle\endcsname\relax
  \providecommand{\doi}[1]{doi: #1}\else
  \providecommand{\doi}{doi: \begingroup \urlstyle{rm}\Url}\fi

\bibitem[Kim et~al.(2017)Kim, Park, and Lee]{kim2017effective}
Hyejeong Kim, Hanwook Park, and Sang~Joon Lee.
\newblock Effective method for drug injection into subcutaneous tissue.
\newblock \emph{Scientific Reports}, 7\penalty0 (1):\penalty0 1--11, 2017.

\bibitem[Ogston-Tuck(2014)]{ogston2014subcutaneous}
Sherri Ogston-Tuck.
\newblock Subcutaneous injection technique: an evidence-based approach.
\newblock \emph{Nursing Standard}, 29\penalty0 (3):\penalty0 53--58, 2014.

\bibitem[Dychter et~al.(2012)Dychter, Gold, and
  Haller]{dychter2012subcutaneous}
Samuel~S Dychter, David~A Gold, and Michael~F Haller.
\newblock Subcutaneous drug delivery: a route to increased safety, patient
  satisfaction, and reduced costs.
\newblock \emph{Journal of Infusion Nursing}, 35\penalty0 (3):\penalty0
  154--160, 2012.

\bibitem[Shapiro(2013)]{shapiro2013use}
Ralph~S Shapiro.
\newblock Why i use subcutaneous immunoglobulin (scig).
\newblock \emph{Journal of Clinical Immunology}, 33\penalty0 (2):\penalty0
  95--98, 2013.

\bibitem[Allen and Ansel(2013)]{allen2013ansel}
Loyd Allen and Howard~C Ansel.
\newblock Ansel's pharmaceutical dosage forms and drug delivery systems.
\newblock \emph{American Journal of Pharmaceutical Education}, 70\penalty0
  (3):\penalty0 X1, 2013.

\bibitem[Prettyman(2005)]{prettyman2005subcutaneous}
Julie Prettyman.
\newblock Subcutaneous or intramuscular? confronting a parenteral
  administration dilemma.
\newblock \emph{Medsurg Nursing}, 14\penalty0 (2):\penalty0 93--99, 2005.

\bibitem[Stoner et~al.(2015)Stoner, Harder, Fallowfield, and
  Jenkins]{stoner2015intravenous}
Kelly~L Stoner, Helena Harder, Lesley~J Fallowfield, and Valerie~A Jenkins.
\newblock Intravenous versus subcutaneous drug administration. which do
  patients prefer? a systematic review.
\newblock \emph{The Patient-Patient-Centered Outcomes Research}, 8\penalty0
  (2):\penalty0 145--153, 2015.

\bibitem[Haller(2007)]{haller2007converting}
Michael~F Haller.
\newblock Converting intravenous dosing to subcutaneous dosing with recombinant
  human hyaluronidase.
\newblock \emph{Pharmaceutical Technology}, 31\penalty0 (10), 2007.

\bibitem[Geerligs et~al.(2008)Geerligs, Peters, Ackermans, Oomens, and
  Baaijens]{geerligs2008linear}
Marion Geerligs, Gerrit W~M Peters, Paul A~J Ackermans, Cees W~J Oomens, and
  Frank Baaijens.
\newblock Linear viscoelastic behavior of subcutaneous adipose tissue.
\newblock \emph{Biorheology}, 45\penalty0 (6):\penalty0 677--688, 2008.

\bibitem[Comley and Fleck(2011)]{comley2011deep}
Kerstyn Comley and Norman Fleck.
\newblock Deep penetration and liquid injection into adipose tissue.
\newblock \emph{Journal of Mechanics of Materials and Structures}, 6\penalty0
  (1):\penalty0 127--140, 2011.

\bibitem[Derler and Gerhardt(2012)]{derler2012tribology}
S~Derler and L-C Gerhardt.
\newblock Tribology of skin: review and analysis of experimental results for
  the friction coefficient of human skin.
\newblock \emph{Tribology Letters}, 45\penalty0 (1):\penalty0 1--27, 2012.

\bibitem[Shrestha and Stoeber(2018)]{shrestha2018fluid}
Pranav Shrestha and Boris Stoeber.
\newblock Fluid absorption by skin tissue during intradermal injections through
  hollow microneedles.
\newblock \emph{Scientific Reports}, 8\penalty0 (1):\penalty0 1--13, 2018.

\bibitem[Shrestha and Stoeber(2020)]{shrestha2020imaging}
Pranav Shrestha and Boris Stoeber.
\newblock Imaging fluid injections into soft biological tissue to extract
  permeability model parameters.
\newblock \emph{Physics of Fluids}, 32\penalty0 (1):\penalty0 011905, 2020.

\bibitem[Barry and Aldis(1992)]{barry1992flow}
S~I Barry and G~K Aldis.
\newblock Flow-induced deformation from pressurized cavities in absorbing
  porous tissues.
\newblock \emph{Bulletin of Mathematical Biology}, 54\penalty0 (6):\penalty0
  977--997, 1992.

\bibitem[Oomens et~al.(1987)Oomens, Van~Campen, and
  Grootenboer]{oomens1987mixture}
C~W~J Oomens, D~H Van~Campen, and H~J Grootenboer.
\newblock A mixture approach to the mechanics of skin.
\newblock \emph{Journal of Biomechanics}, 20\penalty0 (9):\penalty0 877--885,
  1987.

\bibitem[Truesdell and Toupin(1960)]{truesdell1960classical}
Clifford Truesdell and Richard Toupin.
\newblock \emph{The Classical Field Theories}.
\newblock Springer, 1960.

\bibitem[Rajagopal and Tao(1995)]{rajagopal1995mechanics}
Kumbakonam~R Rajagopal and Luoyi Tao.
\newblock \emph{Mechanics of Mixtures}, volume~35.
\newblock World Scientific, 1995.

\bibitem[Truesdell et~al.(2004)Truesdell, Noll, Truesdell, and
  Noll]{truesdell2004non}
Clifford Truesdell, Walter Noll, C~Truesdell, and W~Noll.
\newblock \emph{The Non-linear Field Theories of Mechanics}.
\newblock Springer, 2004.

\bibitem[Rajagopal et~al.(1986)Rajagopal, Wineman, and
  Gandhi]{rajagopal1986boundary}
KR~Rajagopal, AS~Wineman, and M~Gandhi.
\newblock On boundary conditions for a certain class of problems in mixture
  theory.
\newblock \emph{International Journal of Engineering Science}, 24\penalty0
  (8):\penalty0 1453--1463, 1986.

\bibitem[Gandhi et~al.(1987)Gandhi, Rajagopal, and Wineman]{gandhi1987some}
MV~Gandhi, KR~Rajagopal, and AS~Wineman.
\newblock Some nonlinear diffusion problems within the context of the theory of
  interacting continua.
\newblock \emph{International Journal of Engineering Science}, 25\penalty0
  (11-12):\penalty0 1441--1457, 1987.

\bibitem[Wineman and Rajagopal(1992)]{wineman1992shear}
A~Wineman and Kumbakonam~R Rajagopal.
\newblock Shear induced redistribution of fluid within a uniformly swollen
  nonlinear elastic cylinder.
\newblock \emph{International Journal of Engineering Science}, 30\penalty0
  (11):\penalty0 1583--1595, 1992.

\bibitem[Byrne and Preziosi(2003)]{byrne2003modelling}
Helen Byrne and Luigi Preziosi.
\newblock Modelling solid tumour growth using the theory of mixtures.
\newblock \emph{Mathematical Medicine and Biology: A Journal of the IMA},
  20\penalty0 (4):\penalty0 341--366, 2003.

\bibitem[Rajagopal(2007)]{rajagopal2007hierarchy}
Kumbakonam~R Rajagopal.
\newblock On a hierarchy of approximate models for flows of incompressible
  fluids through porous solids.
\newblock \emph{Mathematical Models and Methods in Applied Sciences},
  17\penalty0 (02):\penalty0 215--252, 2007.

\bibitem[Kumar et~al.(2018)Kumar, Dey, and Sekhar]{kumar2018nutrient}
Prakash Kumar, Bibaswan Dey, and G~P~Raja Sekhar.
\newblock Nutrient transport through deformable cylindrical scaffold inside a
  bioreactor: an application to tissue engineering.
\newblock \emph{International Journal of Engineering Science}, 127:\penalty0
  201--216, 2018.

\bibitem[Dey and Raja~Sekhar(2021)]{dey2021mathematical}
Bibaswan Dey and G~P Raja~Sekhar.
\newblock Mathematical modeling of electrokinetic transport through
  endothelial-cell glycocalyx.
\newblock \emph{Physics of Fluids}, 33\penalty0 (8):\penalty0 081902, 2021.

\bibitem[Anguiano et~al.(2022)Anguiano, Masud, and
  Rajagopal]{anguiano2022mixture}
Marcelino Anguiano, Arif Masud, and Kumbakonam~R Rajagopal.
\newblock Mixture model for thermo-chemo-mechanical processes in fluid-infused
  solids.
\newblock \emph{International Journal of Engineering Science}, 174:\penalty0
  103576, 2022.

\bibitem[Byrne et~al.(2003)Byrne, King, McElwain, and Preziosi]{byrne2003two}
Helen~M Byrne, John~R King, D~L~Sean McElwain, and Luigi Preziosi.
\newblock A two-phase model of solid tumour growth.
\newblock \emph{Applied Mathematics Letters}, 16\penalty0 (4):\penalty0
  567--573, 2003.

\bibitem[Breward et~al.(2003)Breward, Byrne, and Lewis]{breward2003multiphase}
Christopher J~W Breward, Helen~M Byrne, and Claire~E Lewis.
\newblock A multiphase model describing vascular tumour growth.
\newblock \emph{Bulletin of Mathematical Biology}, 65\penalty0 (4):\penalty0
  609--640, 2003.

\bibitem[Rajagopal(2010)]{rajagopal2010mechanics}
Kumbakonam~Ramamani Rajagopal.
\newblock Mechanics of liquid mixtures.
\newblock In \emph{Rheology of Complex Fluids}, pages 67--84. Springer, 2010.

\bibitem[Chen et~al.(2001)Chen, Byrne, and King]{chen2001influence}
C~Y Chen, H~M Byrne, and J~R King.
\newblock The influence of growth-induced stress from the surrounding medium on
  the development of multicell spheroids.
\newblock \emph{Journal of Mathematical Biology}, 43\penalty0 (3):\penalty0
  191--220, 2001.

\bibitem[Barry and Aldis(1990)]{barry1990comparison}
S~I Barry and G~K Aldis.
\newblock Comparison of models for flow induced deformation of soft biological
  tissue.
\newblock \emph{Journal of Biomechanics}, 23\penalty0 (7):\penalty0 647--654,
  1990.

\bibitem[Barry et~al.(1991)Barry, Parkerf, and Aldis]{barry1991fluid}
S~I Barry, K~H Parkerf, and G~K Aldis.
\newblock Fluid flow over a thin deformable porous layer.
\newblock \emph{Zeitschrift f{\"u}r Angewandte Mathematik und Physik ZAMP},
  42\penalty0 (5):\penalty0 633--648, 1991.

\bibitem[Barry et~al.(1995)Barry, Aldis, and Mercer]{barry1995injection}
Steven~I Barry, Geoffrey~K Aldis, and Geoffry~N Mercer.
\newblock Injection of fluid into a layer of deformable porous medium.
\newblock \emph{Applied Mechanics Reviews}, 40\penalty0 (10):\penalty0
  722--726, 1995.

\bibitem[Barry et~al.(1997)Barry, Mercer, and Zoppou]{barry1997deformation}
S~I Barry, G~N Mercer, and C~Zoppou.
\newblock Deformation and fluid flow due to a source in a poro-elastic layer.
\newblock \emph{Applied Mathematical Modelling}, 21\penalty0 (11):\penalty0
  681--689, 1997.

\bibitem[Li and Johnson(2009)]{li2009mathematical}
Jiaxu Li and James~D Johnson.
\newblock Mathematical models of subcutaneous injection of insulin analogues: a
  mini-review.
\newblock \emph{Discrete and Continuous Dynamical Systems. Series B},
  12\penalty0 (2):\penalty0 401, 2009.

\bibitem[Mow et~al.(1984)Mow, Holmes, and Lai]{mow1984fluid}
Van~C. Mow, Mark~H Holmes, and W~Michael Lai.
\newblock Fluid transport and mechanical properties of articular cartilage: a
  review.
\newblock \emph{Journal of Biomechanics}, 17\penalty0 (5):\penalty0 377--394,
  1984.

\bibitem[Holmes(1985)]{holmes1985theoretical}
Mark~H Holmes.
\newblock A theoretical analysis for determining the nonlinear hydraulic
  permeability of a soft tissue from a permeation experiment.
\newblock \emph{Bulletin of Mathematical Biology}, 47\penalty0 (5):\penalty0
  669--683, 1985.

\bibitem[Wei et~al.(2003)Wei, Waters, Liu, and Grotberg]{wei2003flow}
H~H Wei, S~L Waters, Shu~Qian Liu, and J~B Grotberg.
\newblock Flow in a wavy-walled channel lined with a poroelastic layer.
\newblock \emph{Journal of Fluid Mechanics}, 492:\penalty0 23--45, 2003.

\bibitem[Dey and Raja~Sekhar(2016{\natexlab{a}})]{dey2016hydrodynamics}
Bibaswan Dey and G~P Raja~Sekhar.
\newblock Hydrodynamics and convection enhanced macromolecular fluid transport
  in soft biological tissues: Application to solid tumor.
\newblock \emph{Journal of Theoretical Biology}, 395:\penalty0 62--86,
  2016{\natexlab{a}}.

\bibitem[Rees and Storesletten(1995)]{rees1995effect}
D~Andrew~S Rees and L~Storesletten.
\newblock The effect of anisotropic permeability on free convective boundary
  layer flow in porous media.
\newblock \emph{Transport in Porous Media}, 19\penalty0 (1):\penalty0 79--92,
  1995.

\bibitem[Payne et~al.(2001)Payne, Rodrigues, and Straughan]{payne2001effect}
L~E Payne, J~F Rodrigues, and B~Straughan.
\newblock Effect of anisotropic permeability on darcy's law.
\newblock \emph{Mathematical Methods in the Applied Sciences}, 24\penalty0
  (6):\penalty0 427--438, 2001.

\bibitem[Karmakar and Raja~Sekhar(2018{\natexlab{a}})]{karmakar2018effect}
Timir Karmakar and G~P Raja~Sekhar.
\newblock Effect of anisotropic permeability on convective flow through a
  porous tube with viscous dissipation effect.
\newblock \emph{Journal of Engineering Mathematics}, 110\penalty0 (1):\penalty0
  15--37, 2018{\natexlab{a}}.

\bibitem[Rajani et~al.(2020)Rajani, Narla, and
  Hemalatha]{rajani2020anisotropic}
D~Rajani, V~K Narla, and K~Hemalatha.
\newblock Anisotropic permeability impact on nanofluid channel flow
  (ch3oh-fe3o4) with convection.
\newblock \emph{Materials Today: Proceedings}, 28:\penalty0 2251--2257, 2020.

\bibitem[Kohr et~al.(2008)Kohr, Raja~Sekhar, and Blake]{kohr2008green}
Mirela Kohr, G~P Raja~Sekhar, and John~R Blake.
\newblock Green's function of the brinkman equation in a 2d anisotropic case.
\newblock \emph{IMA Journal of Applied Mathematics}, 73\penalty0 (2):\penalty0
  374--392, 2008.

\bibitem[Karmakar and Raja~Sekhar(2017)]{karmakar2017note}
Timir Karmakar and G~P Raja~Sekhar.
\newblock A note on flow reversal in a wavy channel filled with anisotropic
  porous material.
\newblock \emph{Proceedings of the Royal Society A: Mathematical, Physical and
  Engineering Sciences}, 473\penalty0 (2203):\penalty0 20170193, 2017.

\bibitem[Reynaud and Quinn(2006)]{reynaud2006anisotropic}
Boris Reynaud and Thomas~M Quinn.
\newblock Anisotropic hydraulic permeability in compressed articular cartilage.
\newblock \emph{Journal of Biomechanics}, 39\penalty0 (1):\penalty0 131--137,
  2006.

\bibitem[Federico and Herzog(2008{\natexlab{a}})]{federico2008anisotropy}
Salvatore Federico and Walter Herzog.
\newblock On the anisotropy and inhomogeneity of permeability in articular
  cartilage.
\newblock \emph{Biomechanics and Modeling in Mechanobiology}, 7\penalty0
  (5):\penalty0 367--378, 2008{\natexlab{a}}.

\bibitem[Iatridis et~al.(1998)Iatridis, Setton, Foster, Rawlins, Weidenbaum,
  and Mow]{iatridis1998degeneration}
James~C Iatridis, Lori~A Setton, Robert~J Foster, Bernard~A Rawlins, Mark
  Weidenbaum, and Van~C Mow.
\newblock Degeneration affects the anisotropic and nonlinear behaviors of human
  anulus fibrosus in compression.
\newblock \emph{Journal of Biomechanics}, 31\penalty0 (6):\penalty0 535--544,
  1998.

\bibitem[Federico and Herzog(2008{\natexlab{b}})]{federico2008permeability}
Salvatore Federico and Walter Herzog.
\newblock On the permeability of fibre-reinforced porous materials.
\newblock \emph{International Journal of Solids and Structures}, 45\penalty0
  (7-8):\penalty0 2160--2172, 2008{\natexlab{b}}.

\bibitem[Holzapfel et~al.(2001)]{holzapfel2001biomechanics}
Gerhard~A Holzapfel et~al.
\newblock Biomechanics of soft tissue.
\newblock \emph{The Handbook of Materials Behavior Models}, 3\penalty0
  (1):\penalty0 1049--1063, 2001.

\bibitem[Shepherd(2018)]{shepherd2018injection}
E~Shepherd.
\newblock Injection technique 2: administering drugs via the subcutaneous
  route.
\newblock \emph{Nursing Times}, 114:\penalty0 55--57, 2018.

\bibitem[Dey and Raja~Sekhar(2016{\natexlab{b}})]{dey2016theoretical}
Bibaswan Dey and G~P Raja~Sekhar.
\newblock A theoretical study on the elastic deformation of cellular phase and
  creation of necrosis due to the convection reaction process inside a
  spherical tumor.
\newblock \emph{International Journal of Biomathematics}, 9\penalty0
  (06):\penalty0 1650095, 2016{\natexlab{b}}.

\bibitem[Barbieri et~al.(2014)Barbieri, Wanat, and Seykora]{barbieri2014skin}
J~S Barbieri, K~Wanat, and J~Seykora.
\newblock Skin: basic structure and function.
\newblock In \emph{Pathobiology of Human Disease: A Dynamic Encyclopedia of
  Disease Mechanisms}, pages 1134--1144. Elsevier, 2014.

\bibitem[Savatorova and Rajagopal(2011)]{savatorova2011homogenization}
VL~Savatorova and KR~Rajagopal.
\newblock Homogenization of a generalization of brinkman's equation for the
  flow of a fluid with pressure dependent viscosity through a rigid porous
  solid.
\newblock \emph{ZAMM-Journal of Applied Mathematics and Mechanics/Zeitschrift
  f{\"u}r Angewandte Mathematik und Mechanik}, 91\penalty0 (8):\penalty0
  630--648, 2011.

\bibitem[Dey and Raja~Sekhar(2014)]{dey2014mass}
Bibaswan Dey and G~P Raja~Sekhar.
\newblock Mass transfer and species separation due to oscillatory flow in a
  brinkman medium.
\newblock \emph{International Journal of Engineering Science}, 74:\penalty0
  35--54, 2014.

\bibitem[Degan et~al.(2002)Degan, Zohoun, and Vasseur]{degan2002forced}
G~Degan, S~Zohoun, and P~Vasseur.
\newblock Forced convection in horizontal porous channels with hydrodynamic
  anisotropy.
\newblock \emph{International Journal of Heat and Mass Transfer}, 45\penalty0
  (15):\penalty0 3181--3188, 2002.

\bibitem[Stenkula and Erlanson-Albertsson(2018)]{stenkula2018adipose}
Karin~G Stenkula and Charlotte Erlanson-Albertsson.
\newblock Adipose cell size: importance in health and disease.
\newblock \emph{American Journal of Physiology-Regulatory, Integrative and
  Comparative Physiology}, 315\penalty0 (2):\penalty0 R284--R295, 2018.

\bibitem[Prasad and Rajagopal(2006)]{prasad2006diffusion}
Sharat~C Prasad and Kumbakonam~R Rajagopal.
\newblock On the diffusion of fluids through solids undergoing large
  deformations.
\newblock \emph{Mathematics and Mechanics of Solids}, 11\penalty0 (3):\penalty0
  291--305, 2006.

\bibitem[Beavers and Joseph(1967)]{beavers1967boundary}
Gordon~S Beavers and Daniel~D Joseph.
\newblock Boundary conditions at a naturally permeable wall.
\newblock \emph{Journal of Fluid Mechanics}, 30\penalty0 (1):\penalty0
  197--207, 1967.

\bibitem[Jones(1973)]{jones1973low}
I~P Jones.
\newblock Low reynolds number flow past a porous spherical shell.
\newblock In \emph{Mathematical Proceedings of the Cambridge Philosophical
  Society}, volume~73, pages 231--238. Cambridge University Press, 1973.

\bibitem[Karmakar and Raja~Sekhar(2018{\natexlab{b}})]{karmakar2018squeeze}
Timir Karmakar and G~P Raja~Sekhar.
\newblock Squeeze-film flow between a flat impermeable bearing and an
  anisotropic porous bed.
\newblock \emph{Physics of Fluids}, 30\penalty0 (4):\penalty0 043604,
  2018{\natexlab{b}}.

\bibitem[Hill and Straughan(2008)]{hill2008poiseuille}
Antony~A Hill and Brian Straughan.
\newblock Poiseuille flow in a fluid overlying a porous medium.
\newblock \emph{Journal of Fluid Mechanics}, 603:\penalty0 137--149, 2008.

\bibitem[Tsangaris and Leiter(1984)]{tsangaris1984laminar}
S~Tsangaris and E~Leiter.
\newblock On laminar steady flow in sinusoidal channels.
\newblock \emph{Journal of Engineering Mathematics}, 18\penalty0 (2):\penalty0
  89--103, 1984.

\bibitem[Khor et~al.(1991)Khor, Bozigian, and Mayersohn]{khor1991potential}
SP~Khor, Haig Bozigian, and Michael Mayersohn.
\newblock Potential error in the measurement of tissue to blood distribution
  coefficients in physiological pharmacokinetic modeling. residual tissue
  blood. ii. distribution of phencyclidine in the rat.
\newblock \emph{Drug Metabolism and Disposition}, 19\penalty0 (2):\penalty0
  486--490, 1991.

\bibitem[Truskey et~al.(2010)Truskey, Yuan, and Katz]{truskey2010transport}
George~A Truskey, Fan Yuan, and David~F Katz.
\newblock \emph{Transport phenomena in biological systems}, volume~2.
\newblock Pearson New Jersey, 2010.

\bibitem[Comley and Fleck(2010)]{comley2010micromechanical}
Kerstyn Comley and Norman~A Fleck.
\newblock A micromechanical model for the young’s modulus of adipose tissue.
\newblock \emph{International Journal of Solids and Structures}, 47\penalty0
  (21):\penalty0 2982--2990, 2010.

\bibitem[Yao et~al.(2012)Yao, Li, and Ding]{yao2012interstitial}
Wei Yao, Yabei Li, and Guanghong Ding.
\newblock Interstitial fluid flow: the mechanical environment of cells and
  foundation of meridians.
\newblock \emph{Evidence-Based Complementary and Alternative Medicine}, 2012,
  2012.

\bibitem[Mueller et~al.(2005)Mueller, Zou, and Lott]{mueller2005pressure}
Michael~J Mueller, Dequan Zou, and Donovan~J Lott.
\newblock “pressure gradient” as an indicator of plantar skin injury.
\newblock \emph{Diabetes Care}, 28\penalty0 (12):\penalty0 2908--2912, 2005.

\bibitem[Goossens(2009)]{goossens2009fundamentals}
R~H~M Goossens.
\newblock Fundamentals of pressure, shear and friction and their effects on the
  human body at supported postures.
\newblock In \emph{Bioengineering Research of Chronic Wounds}, pages 1--30.
  Springer, 2009.

\end{thebibliography}

\end{document}